\documentclass[11pt,a4paper]{article}
\usepackage{lineno}
\usepackage{longtable}
\usepackage[utf8]{inputenc}
\usepackage{amsmath}
\usepackage{amsthm}
\usepackage[bitstream-charter]{mathdesign}
\usepackage{siunitx}
\usepackage{adjustbox}
\newtheorem{theorem}{Theorem}
\usepackage{xcolor}
\usepackage{graphicx}
\usepackage[margin=1in]{geometry}
\usepackage{float}
\usepackage{verbatimbox}
\usepackage{booktabs}
\usepackage{algorithm}
\usepackage{algorithmic}
\usepackage{authblk}
\usepackage{lipsum}

\usepackage{caption} 
\usepackage{subcaption} 
\usepackage[T1]{fontenc}
\theoremstyle{definition}

\renewcommand{\arraystretch}{1.5}

\makeatletter
\def\thm@space@setup{\thm@preskip=1.2\parskip \thm@postskip=0pt}
\makeatother
\usepackage{tikz}
\usetikzlibrary{arrows}
\usepackage[toc,page]{appendix}
\usepackage{apacite}

\usepackage{xurl} 
\usepackage[authoryear, round]{natbib}

\usepackage{hyperref}\hypersetup{
  colorlinks=true,
  citecolor=blue,
  linkcolor=red,
  filecolor=blue,
  urlcolor=blue,
}

\begin{document}
\pagenumbering{gobble} 

%

\section*{Highlights}

\vspace{1em}
\begin{itemize}
	\item An Ordinary differential equation model for a wound healing process is adapted and extended to include systemic factors and chronic wounds. 
	
	\item Effects of oxygen, pathogen activity and inflammatory activity on wound healing are investigated. 
	
	\item Positivity of solutions and exsistence of equilibria are proved. 
	
	\item A global sensitivity analysis is conducted to identify the parameters with the greatest influence on the wound healing process.

\end{itemize}

\newpage

\clearpage 
\pagenumbering{arabic} 
\setcounter{page}{1} 

\title{A Mathematical Model of Wound Healing: Effects of Local and Systemic Factors}


\author{
	Alinafe Maenje$^{1}$\thanks{Corresponding author: \texttt{alinafe@aims.ac.za}}, and Joseph Malinzi$^{2,3}$ \\
	\small
	$^{1}$Department of Mathematics and Statistics, Malawi University of Business and Applied Sciences, Malawi \\
	\small
	$^{2}$ Department of Mathematics, University of Eswatini, Private Bag 4, Kwaluseni, Eswatini\\
	$^{3}$ Institute of Systems Science, Durban University of Technology, Durban 4000, South Africa
}

\date{}
\maketitle

\begin{abstract}

 This study presents a mathematical model formulated as a system of first-order non-linear ordinary differential equations, aimed at examining the effects of different factors, classified as local and systemic factors on a wound healing process. Specifically, the model incorporates pathogens, inflammatory cells, fibroblast cells, collagen, and the initial wound size. The variables utilized in the model lack specific units since they represent combined responses of different cell types. As such, these variables serve as indicators of the relative behaviors of these cellular populations. The feasibility of the model solution is examined through the demonstration of its positivity. A threshold, denoted $R_w$, is established to determine the conditions necessary for the existence of a wound-free equilibrium. Sensitivity analysis is carried out to determine the contribution of parameters to the behavior of fibroblast cells and $R_w$, which are crucial to the wound healing process. Through numerical simulations, it is demonstrated that factors such as oxygen levels, pathogen activity, inflammatory cell activity, age, smoking, alcoholism, and stress exert considerable influence on the behaviors of these cells, thereby contributing to delayed wound healing and the occurrence of chronic wounds.
\vspace{2em}

\noindent \textbf{Keywords:} Mathematical modelling, Ordinary differential equations, wound healing, acute wounds, chronic wounds.
\end{abstract}


\section{Introduction}

It is inevitable that human beings suffer from wounds that can be acute or chronic depending on the time they take to heal. Acute wounds are those that heal within a normal wound healing time frame, around 3 weeks, whereas chronic wounds are those that fail to heal in a normal time frame \citep{frykbergrobert2015challenges}. While some wounds such as pressure ulcers, diabetic foot ulcers, and venous leg ulcers, are classified as chronic, failure of an acute wound to heal in the normal time frame also leads to chronic wounds. These acute wounds may get stuck at one stage or fail to respond to treatment. This may be due to different diseases such as diabetes and infection. Individual chronic wound care costs approximately 9060 Euros a year on average in Germany, the total annual costs of chronic wound care is around 4.5 to 5.3 billion Euros in The United Kingdom \citep{velivckovic2022cost} and the United States spends USD28.1 to USD96.8 billion yearly on average \citep{sen2021human}. In Africa, little is known about care costs for both acute and chronic wounds. In Kenya, the cost for wound care is approximately USD2.25 per week according to a study by \cite{odhiambo2019wound} which concluded that the costs are affordable but the care provided is not up to standard. However, the study was only conducted in one hospital in Kenya. In Nigeria, a study by  \cite{ogundeji2022modelling} approximates wound care cost to be around 36,922 Naira per week, yet, the study is inconclusive on the question of affordability for the household with average monthly income. According to \cite{varela2017untreated}, people in African countries like Malawi, Uganda, and Sierra Leone tend to fail to get timely or proper medical care due to poverty. Failure of a wound to heal normally can lead to excessive scarring, paralysis, and death through septic shock \citep{martin2005inflammatory}.

The wound healing process is affected by a number of factors. \cite{guo2010factors} classify these factors as local and systemic factors. Local factors are defined as those that have a direct impact on the features of the wound \citep{guo2010factors}. According to \cite{rodriguez2008role}, oxygen is one of the local factors and important element in the wound healing process. It helps in inhibiting wound infection, promotes angiogenesis, increases fibroblast proliferation and collagen production, and promotes wound contraction. Wound infection is another local factor affecting the healing process. \cite{edwards2004bacteria} classify wounds according to the three states of infections. The contamination state occurs when a wound contains non-replicating organisms. A wound with reproducing organisms is classified as being in a colonisation state, while a wound in an intermediate condition is referred to as being in a local infection or critical colonisation state. Systemic factors are those that put an individual in a certain health or disease state that affect wound healing \citep{guo2010factors}. Systemic factors include age, stress, diabetes, obesity, medication, sex hormones in aged individuals, alcohol consumption, smoking and nutrition.  \cite{guo2010factors} conclude that individuals above the age of 60 experience a prolonged wound healing process. Among these aged people, the wound healing period also differs depending on gender with aged males having a longer wound healing period than aged females simply due to sex hormones. A study by \cite{vileikyte2007stress} on stress and wound healing finds that there is an association between stress and abnormal wound healing process. Diseases such as diabetes and obesity are also found to be associated with prolonged wound healing.  \cite{guo2010factors} also find a prolonged wound healing process in smoking and alcoholic individuals.  \cite{arnold2006nutrition} recommend foods rich in carbohydrates, protein, and amino acids to help in the process of wound healing citing the lack of certain food as a systemic factor in the wound healing process. 

To date, mathematical modelling has facilitated the study of various phenomena in the wound healing process by providing a robust framework to describe these processes. \cite{Hay2024} recently developed a model which describes the wound healing process of corals of the Pocillopora damicornis family, which is defined by a system of ordinary differential equations with hill-type functions and constants. The model was further validated using experimental data for coral reefs and the results showed a positive correlation between experimental data and the model prediction. In humans, quite a number of models have been developed and studied, for example \cite{reynolds2006reduced,menke2010silico, almeida2011mathematical, segal2012differential, cooper2015modeling, torres2019identifying} and some references therein. Some of these models, \cite{almeida2011mathematical,cooper2015modeling, torres2019identifying} have neglected to account for the crucial role of tissue oxygenation in wound healing, while others have overlooked the inflammatory response, which is known to be vital in this context. Although models proposed in \cite{reynolds2006reduced, menke2010silico,  segal2012differential} offer valuable insights into the wound healing process, they have overlooked the analysis and effect of other local factors and systemic factors. Consequently, this study aims to build upon previous studies, especially that of \cite{segal2012differential} to provide updated insights, mathematical analysis and sensitivity analysis to show important parameters in the wound healing process and provide the impact of local and systemic factors on cell activities and to the wound healing process.

\section{The Model}

We adapt a model by \cite{segal2012differential} and we extend it to include systemic factors and chronic wounds. A normal wound healing process has four main stages namely: hemostasis, inflammation, proliferation, and remodelling (maturation) stages \citep{guo2010factors}. The hemostasis stage involves blood clotting, constriction of blood vessels, blood changes to a semi-fluid state, and trapping platelets with threads of fibrin, a process called coagulation. This process is a result of collagen coming together with platelets in the presence of an enzyme called thrombin and it lasts 0 to 24 hours \citep{guo2010factors}. The second stage is the inflammation stage which involves localized swelling caused by leaking of discharges, helping to remove damaged cells, tissues, and bacteria. Normally, this stage takes about 1 to 7 days from injury although it can be delayed if there is further infection \citep{bertone1989principles}. The wound then goes into the proliferation stage, a stage which involves the building up of new tissue and blood vessels (angiogenesis) \citep{velnar2009wound}. Fibroblasts cause the wound to contract, pulling wound edges together, starting around day 4 to around day 24. Finally, the maturation stage (also known as the remodeling stage) commences approximately on day 14 and extends up to about 2 years. This stage involves the remodelling of collagen and full wound closure,  encompassing collagen remodeling and complete wound closure, along with the elimination of surplus cells involved in the closure process \citep{flegg2012wound}.

We consider the presence of the wound to be the driver of the wound healing process with collagen being the main protein component in the process. The model tracks the relative activity of collagen in the healing wound. The presence of the initial wound ($WS_0$) initiates the activation of inflammatory cells due to the presence of pathogens on the wound. This inflammation will activate more inflammation. The activated inflammatory cells, which include several cell types and their mediators, destroy fibroblast cells which produce collagen for tissue remodelling \citep{segal2012differential}.

\begin{figure}[]
	\centering
	\includegraphics[width=0.9\linewidth]{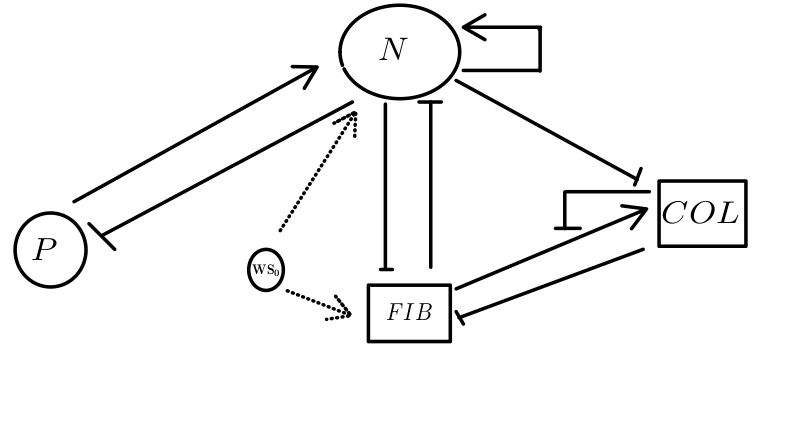}
	\caption{\begin{small}
			Diagrammatic representation of the model. The initial wound size, $WS_0$, is used in the collagen accumulation model. This leads to the recruitment of pathogens (P). Inflammatory cells (N) and fibroblast cells (which produce collagen) are produced by the surrounding healthy tissues. Up regulation is shown by arrows while down regulation (inhibition) is represented by bars.\end{small}}
	\label{fig:model}
\end{figure}

Activated inflammatory cells ($N$) destroy fibroblast cells ($FIB$) in a non-linear process. Activated inflammatory cells also destroy collagen and pathogens. However, the destruction of pathogens by inflammatory cells is dependent on oxygen levels and the inhibition of inflammatory cells by fibroblast cells. Inflammatory cells are destroyed by active fibroblast cells. Fibroblast cells are responsible for the production of collagen in a non-linear process subject to contact inhibition. However, less collagen is needed when the wound is going into the remodelling stage so that active fibroblast cells will start to degrade the available collagen.

Model parameters are provided in Table \ref{long1} and variables in Table \ref{long2}. The relative activity of each variable is considered since most variables are a combination of different cells. 

\begin{center}
	\begin{longtable}{l l }
		\caption[Short caption]{Model variables.} \label{long2} \\
		
		\hline
		\multicolumn{1}{c}{\textbf{Variable}} & 
		\multicolumn{1}{c}{\textbf{Description}} \\ \hline
		\endfirsthead	
		\endfoot	
		\hline
		\endlastfoot			
		$FIB_{\text{p}}$ & Proliferating fibroblast cells\\
		$FIB_{\text{m}}$ & Migrating fibroblast cells\\
		$FIB_{\text{a}}$ & Activate fibroblast cells\\
		$N$ & Relative activity of Inflammatory cells and other mediators\\
		$P$ & Pathogen concentration on the wound\\
		$COL$ & Percentage space filled with collagen\\
		$WS$ & Wound size\\	 \hline		
	\end{longtable}
\end{center}

\begin{center}
	\begin{longtable}{l l l}
		\caption[Short caption]{Model Parameters \citep{reynolds2006reduced, segal2012differential}} \label{long1} \\
		
		\hline
		\multicolumn{1}{c}{\textbf{Parameter}} & 
		\multicolumn{1}{c}{\textbf{Description}} & 
		\multicolumn{1}{c}{\textbf{Baseline}} \\ \hline
		\endfirsthead	
		\endfoot	
		\hline
		\endlastfoot		
		$\mu_{\text{fib}}$ & Fibroblasts decay rate & $0.1/ \text{day}$ \\
		$p_{\text{fib}x}$  & Reproduction rate of fibroblasts & $\text{m}:0.6, \text{p}:0.4, \text{a}:0.3/ \text{day}$ \\
		$d_\text{f}$ & Transition rate of fibroblasts & $0.3/ \text{day}$ \\
		$k_{\text{fn}x}$ &Rate of destruction of fibroblast $x~$ by $N$ & $\text{m}:0.5, \text{p}:0.4, \text{a}:0.3/ \text{day}$ \\
		$x_{\text{fn}}$ & Hill constant for destruction of fibroblasts by $N$ & $0.6~$ N-units \\
		$p_{\text{bl}}$ & Baseline reproduction rate being non-zero when \\
		\nonumber& $COL \cdot FIB > 0.01$ & $0.08~\text{F-units} / \text{day}$\\
		$k_{\text{cf}}$ & Rate of production of collagen by $FIB_a$ & $1/\text{day}$\\	$x_{\text{cf}}$ & Hill constant for  production of collagen by $FIB_a$ & $4.2~$ F-units\\
		$h_{\text{fc}}$ & Inhibition exponent for contact inhibition in \\
		\nonumber & $COL~$ production & $1$ \\
		$k_{\text{cn}}$ & Rate of destruction of $COL~$ by $N$  & $2/N-\text{units}/\text{day}$ \\
		$k_{\text{cfr}}$ & Rate of degradation of collagen by $FIB_a$ & $500/\text{day}$\\	$x_{\text{cfr}}$ &Hill constant for degradation of collagen by $FIB_a$ & $5~\text{F-units}$\\
		$N_{\text{crit}}$ & Maximum of $N~$ before the wound starts remodelling &$0.01~N\text{-units}$ \\
		$k_{\text{np}}$ & Rate of activation of $N~$ by $P$ & $0.5/\text{P-units}/ \text{day}$ \\
		$k_{\text{nn}}$ & Rate of activation of $N~$ by immune mediators & $3/ N^3\text{-units}/ \text{day}$\\
		$k_{\text{nw}}$ & Rate of activation of $N~$ by wound size & $2/ \text{day}$\\
		$s_{\text{nr}}$ & Source rate of resting inflammatory cells & $2N~\text{-units} / \text{day}$\\
		$\mu_{\text{nr}}$ & Rate of decay of inflammatory cells & $2.88/ \text{day}$\\
		$k_{\text{nf}}$ & Rate of at which $FIB_a$ destroys $N$ & $0.1/  \text{F-units}/ \text{day}$\\
		$\mu_n$ &Natural death rate of $N$  & $1.2/ \text{day}$\\
		$n_{\infty}$ &Inhibition constant of $FIB_x~$ reproduction/diffusion \\
		\nonumber & by  N & $0.6~\text{N-units}$\\
		$c_{\infty}$ & Contact inhibition constant of $FIB_x~$ reproduction \\
		\nonumber & by collagen  & $0.5$\\
		$h_c$ & Inhibition exponent for contact inhibition of $FIB_x~$\\	\nonumber &reproduction by collagen  & $4$\\
		$c_{ f\infty}$ & Inhibition constant for contact inhibition in \\
		\nonumber &collagen production   & $0.8~\text{COL-units}$\\
		$c_{ fr\infty}$ & Inhibition constant for contact inhibition in \\
		\nonumber &collagen degradation   &  $1.8~\text{COL-units}$\\
		$F_{\infty}$ & Inhibition constant for inflammation inhibition $N~$ by \\
		\nonumber &active fibroblast   & $6~\text{F-units}$\\
		$P_{\infty}$ & Maximum population of pathogens  & $20  \times 10^6/ \text{cc}$\\
		$k_{pgO}$ & Pathogen growth under normal oxygen & $0.55/ \text{day}$\\
		$ \beta_{\text{p}}$ & Maximum growth rate increase of pathogens due to\\
		\nonumber &oxygen reduction & $0.3/ \text{day}$\\
		$O_{\text{crit}}$ & Normal oxygen levels & $25$\\
		$k_{\text{pm}}$ & Rate of destruction of pathogens by immune \\
		\nonumber &mediators M & $0.6/ \text{M-units}/ \text{day}$\\
		$s_{\text{m}}$ & Source of background immune mediators & $0.12/ \text{M-units}/ \text{day}$\\
		$\mu_{\text{m}}$ & death rate of immune mediators & $0.048/ \text{day}$\\
		$k_{\text{mp}}$ & Rate of activation of $M~$ by $P$ & $0.108 \text{P-units}/ \text{day}$\\
		$k_{\text{pn}}$ & Rate of destruction of $P~$ by $N$ & $0.2/\text{N-units}/ \text{day}$\\
		$w_{\text{sg}}$ & Effect of less oxygen on the wound size & $0.6$\\
		$T_0$ & Initial temperature on the wound site & $\SI{37}{\degreeCelsius}$\\
		$T$ & Final temperature on the wound site & $\SI{40}{\degreeCelsius}$\\ \hline
	\end{longtable}
\end{center} 

Fibroblasts are a common cell type that produces collagen protein used to maintain a structural framework for many tissues and in a healing wound. They are raised as a response to inflammation. These fibroblast cells divide (proliferate), move to the wound (migrating), and eventually become activated and create collagen. These fibroblast cells have a multiplication rate, $\text{p}_{\text{fibx}}$ and death rate $-\mu_{\text{fibx}}$. Fibroblasts are available in a normal skin for immediate response to a damaged tissue cell, hence there is an initial value of the proliferating fibroblast cell, $FIB_p=10.~$ Upon the occurrence of the wound, these fibroblasts diffuse to the wound at the rate $\text{d}_f,~$ reducing the initial population of proliferating fibroblasts. Fibroblasts remain non-zero, even if the wound has healed. Hence, another term $\text{p}_{\text{bl}}\;$ is introduced in the equation for $FIB_a$ \citep{menke2010virginia}. Inflammatory cells destroy fibroblast cells. The ability of inflammatory cells to destroy these fibroblast cells has a maximum rate, $\text{k}_\text{fnx},\;$ for each type of fibroblast cells ($\text{x}$). This dynamic is modelled using a Hill function, $f_H(x,V) = (x / (V+x)) \label{8},$ where $x\;$ represents inflammatory cells and $V\;$ is the Hill constant for the destruction of specific type of fibroblast cells destruction by inflammatory cells. Apart from inflammatory cells destroying fibroblast cells, there are two types of inhibition. Cell-mediated inhibition, where one cell suppress the activity of the other cell and contact inhibition. Contact inhibition occurs when the activity of a cell is suppressed due to collagen accumulation on the wound. All the two types of inhibition are non-linear processes and are modelled using different Hill-type functions. Cell-mediated inhibition depends on the cell being inhibited (x), inhibitor (V), and the inhibition constant for inhibition, which is dependent on the cell causing inhibition ($\text{V}_\infty$). The Hill function modelling this process has the form: \begin{equation}
	\label{cont}  f_i(x,V,V_\infty)= \frac{x}{1+(V/V_\infty)^2} 
\end{equation}
The Hill function, which models cell-mediated inhibition, multiplies the proliferation and migration terms. Contact inhibition of fibroblast cells has the form: \begin{align*}f_c(COL,c_\infty,h_\infty) = \frac{1}{1+(COL/c_\infty)^{h_\infty}},  \end{align*} \label{10} where $c_\infty~$ is the contact inhibition constant of proliferating fibroblast cells by collagen and $h_\infty~$ is the inhibition exponent for contact inhibition \citep{segal2012differential}. The availability of more collagen matrix means there are more active fibroblasts, the body then regulates the production of more fibroblasts. This means the contact inhibition term will affect the proliferation of fibroblast cells, hence the proliferation terms will be multiplied by the contact inhibition term. When fibroblast cells are exposed to inflammation, their ability to multiply and migrate is also impaired, a process that follows a Hill function: 
\begin{align*}
	f_H(x,V) &=\frac{x}{V+x}, ~\text{where} ~~ 
	f_H(N,x_{\text{fn}}) =\frac{N}{x_{\text{fn}}+N}.
\end{align*}

These dynamics give the equations (\ref{1})-(\ref{3}): 

\begin{align}	
	\frac{dFIB_\text{p}}{dt} &= -\mu_{\text{fib}} FIB_\text{p}+ p_{\text{fibp}} f_i(FIB_\text{p}, N, n_\infty)f_c(COL, c_\infty, 4)- \nonumber 
	d_f f_i(FIB_\text{p}, N, n_\infty)-\\
	& ~~~~~k_{\text{fnp}} f_H(N,x_{\text{fn}})FIB_\text{p}, \label{1} \\
	\frac{dFIB_\text{m}}{dt} &= -\mu_{\text{fib}} FIB_\text{m}+ p_{\text{fibm}} f_i(FIB_\text{m}, N, n_\infty)f_c(COL, c_\infty, 4)-  d_f f_i(FIB_\text{m}, N, n_\infty) \nonumber \\ & 
	~~~~+ d_f f_i(FIB_\text{p}, N, n_\infty)-k_{\text{fnm}} f_H(N,x_{\text{fn}})FIB_\text{m}, \label{2}\\
	\frac{dFIB_\text{a}}{dt} &= -\mu_{\text{fib}} FIB_\text{a}+ p_{\text{fiba}} f_i(FIB_\text{a}, N, n_\infty)f_c(COL, c_\infty, 4)+d_f f_i(FIB_\text{m}, N, n_\infty)- \nonumber \\ &~~~~~k_{\text{fna}}f_H(N,x_{\text{fn}})FIB_\text{a} + p_{\text{bl}} \label{3}.
\end{align}

Collagen is produced by active fibroblast cells at the rate $k_{\text{cf}}$. This process is non-linear and is modelled by a Hill-type term $ k_{\text{cf}} f_H(FIB_a,x_\text{cf}),~$ where
\begin{align*}
	f_H(FIB_a,x_\text{cf})=\frac{FIB_a}{ FIB_a+x_\text{cf}}.  \end{align*} \label{11} As the wound heals, less collagen is needed, so the production term will also be multiplied by the term for contact inhibition,\begin{align*} 
	f_c (COL, c_{f\infty}, h_{\text{fc}})=\frac{1}{1+(COL/c_{f\infty})^{h_{\text{fc}}}}. \end{align*} \label{12} Collagen is also destroyed by inflammatory cells at a rate $-k_{\text{cn}}$ and this gives $-k_{\text{cn}}NCOL. $ Under the effect of body temperature, this term becomes $-k_{\text{cn}}NCOL e^{-(T-T_0)^2}.\label{temp1}$ According to \cite{segal2012differential}, fibroblast cells degrade collagen after the wound has healed. This procedure is crucial because it prevents immoderate scarring. This is a non-linear and non-smooth process, so a Hill-type function and a Heaviside function are employed for non-linearity and to provide a smooth approximation of the process respectively. The degradation is at a rate $k_{\text{cfr}},~$ with a Hill-function: \begin{align*}
	f_H(FIB_a,x_\text{cf})=\frac{FIB_a}{ FIB_a+x_\text{cf}}.\end{align*} \label{13} A Heaviside function is applied to make the approximation of this degradation process smooth. For active fibroblast cells to start degrading collagen, there has to be little or no inflammation. Hence, the Heaviside function is a function of inflammatory cells because the absence of inflammatory cells is a clear indication of a healing wound and that means collagen degradation can start. This gives the Heaviside function, \begin{align*}
	s_h(N_{\text{crit}-N}) = \frac{1}{1+e^{-50({N_\text{crit}}-N)}}, 
\end{align*} \label{14} which multiplies the deterioration term. This degradation term is also subject to contact inhibition. However, in this case, the contact inhibition is being applied on the maximum levels of collagen, and hence we will have: \begin{align*}
	(1-f_c(COL, c_{fr\infty}, 12))= 1- \frac{1}{1+(COL/c_{f\infty})^{h_{\text{fc}}}}.
\end{align*} \label{15} 
This gives the collagen equation as:    
\begin{align}
	\frac{dCOL}{dt} &= k_{\text{cf}} f_H(FIB_a,x_\text{cf})f_c (COL, c_{f\infty}, h_{\text{fc}})- k_{\text{cn}}NCOL e^{-(T-T_0)^2}- \nonumber \\ &
	~~~k_{\text{cfr}}f_H(FIB_a,x_{\text{cfr}})s_h(N_{\text{crit}}-N)(1-f_c(COL, c_{fr\infty}, 12)). \label{4}
\end{align}

There are different types of cells produced by the body that helps fight pathogens on the wound, neutrophils and eosinophils being the most active ones. These cells are released by the body to the wound, in response to pathogens and other inflammatory cells, in a resting state, and once activated, they produce mediators, cytokines and chemokines being the common ones, which helps clearing the wound and prevents further tissue damage \citep{davis_smart_2016}. The activated inflammatory cells and their mediators are modelled by $N$. To model this, a constant source rate of resting inflammatory cells $s_{\text{nr}}~$ is considered. These resting cells are activated at a rate $R,\;$ which is dependent on the pathogen population, amount of pro-inflammatory mediators produced in response to the pathogens and the wound size. Hence $R=k_{np}P+k_{nn} N^3+k_{nw}WS \label{16} $. According to  \cite{segal2012differential}, the pro-inflammatory mediators are exponentiated to describe their accumulation to produce a positive effect.  \cite{reynolds2006reduced} establishes that this activation process occurs so fast that it becomes difficult to track resting inflammatory cells, and the activation process simplifies to $(s_{\text{nr}} R/\mu_{\text{nr}}+R)$ where $\mu_{\text{nr}}\;$ is the decay rate of resting inflammatory cells. However, this activation process is inhibited by active fibroblast cells \citep{buckley2001fibroblasts}. Using the function that models cell-mediated inhibition (\ref{cont}), the activation process with inhibition is defined as:\begin{align*}
	R_i= f_i(R,FIB_a,F_\infty) = \frac{R}{1+(FIB_a/F_\infty)^2} = \frac{k_{np}P+k_{nn} N^3+k_{nw}WS}{1+(FIB_a/F_\infty)^2},
\end{align*} \label{17} 
where $F_\infty,$ is the inhibition constant for inhibition of inflammation by active fibroblasts. Fibroblast cells also destroy inflammatory cells, at the rate $-k_{\text{nf}},$ raising normal temperature levels at the wound site, giving $-k_{\text{nf}}FIB_aNe^{-(T-T_0)^2},\label{temp2}$ and there is a natural death rate $-\mu_n\;$ of inflammatory cells \citep{menke2007impaired}, and these dynamics yield the equation:         
\begin{align}
	\frac{dN}{dt}= \frac{s_{\text{nr}}R_i}{\mu_{\text{nr}}+R_i}-k_{\text{nf}}FIB_aNe^{-(T-T_0)^2}-\mu_n N \label{5}.
\end{align}\\

Pathogen growth is modelled logistically. This growth rate is heavily dependent on the oxygenation of the wound. Considering the critical oxygen level to be 25, oxygen levels above critical levels will fix the rate of growth of pathogens, while the pathogen growth rate will increase in less oxygenated wounds \citep{menke2010silico, segal2012differential}. Hence, we have the pathogen growth rate as a function of oxygen defined as:
\[ k_\text{pg} (O_2)= \begin{cases} 
	k_\text{pg0}, & \text{if}~ O_2 \geq O_\text{crit} \\
	
	k_\text{pg0} + \beta_{\text{p}}(1-\frac{O_2}{O_\text{crit}}) & \text{if}~ O_2 < O_\text{crit}.
\end{cases} \label{o2andpath} 
\]
According to  \cite{menke2010virginia}, the pathogen growth term is also dependent on the population of pathogens $P\;$ and the maximum number of the pathogens that can be sustained. These dynamics yield the pathogen growth term as: \begin{align*}
	k_{\text{pg}}(O_2)P \bigg( 1-\frac{P}{P_\infty}\bigg).
\end{align*} \label{19} Apart from immune mediators helping in cleaning the wound, they also help in directly fighting the pathogens. According to \cite{reynolds2006reduced}, this dynamic is modelled as:\begin{align*}
	- \frac{k_\text{pm} s_\text{m}P}{\mu_{\text{m}}+k_{\text{mp}}P},	
\end{align*} \label{20} where $-k_{\text{pm}}\;$ is the rate of destruction of pathogens by these mediators, $s_{\text{m}}\;$ is the constant source of the mediators and $\mu_{\text{m}}\;$ is the natural death rate of the mediators. Activated inflammatory cells destroy pathogens at the rate $k_\text{pn}$. However, the activity of these inflammatory cells is inhibited by active fibroblasts. This dynamic gives the third term $-k_{\text{pn}}P f_i(N, FIB_a, F_\infty)\;$. This activity, however is dependent on oxygen levels. Low oxygen levels will increase the number of pathogens and reduce the activity of inflammatory cells and sufficient oxygen levels will increase the activity of inflammatory cells and reduce the increase of pathogen population. This leads to an addition of a function of oxygen in the term and gives: $-k_{\text{pn}}P f_i(N, FIB_a, F_\infty) (1-g(O_2),\;$  where 
\begin{equation*} \label{oxy}	 
	g(O_2)= 1-\frac{1.1}{1+ 0.1e^{-0.3(O_2-O_{\text{crit}})}}.
\end{equation*}
This gives the pathogen activity equation as:
\begin{align}
	\label{6}
	\frac{dP}{dt}= k_{\text{pg}}(O_2)P \bigg( 1-\frac{P}{P_\infty}\bigg) -\frac{k_\text{pm} s_\text{m}P}{\mu_{\text{m}}+k_{\text{mp}} P}-k_{\text{pn}}P f_i(N, FIB_a, F_\infty) (1-g(O_2)).
\end{align}

According to \cite{segal2012differential}, collagen deposition and oxygenation have a direct effect on wound size. The non-negative variable $WS$ in the model gives wound size. It is also important to note that $COL$ can take values greater than one assuming excessive collagen production and abnormal remodelling We use $w_{\text{sg}}\;$ to account for the effect of altered oxygen levels on the wound size and these dynamics give the wound size equation as:
\begin{align}
	WS=\text{max} \bigg( (1-COL) WS_0 \frac{1}{1-w_{\text{sg}} g(O_2)}, 0\bigg). \label{7}
\end{align}

Combining the equations discussed above, we have the system of equations (\ref{2.10.1})- (\ref{2.10.6}):
\begin{align}	
	\frac{dFIB_\text{p}}{dt} &= -\mu_{\text{fib}} FIB_\text{p}+ p_{\text{fibp}} f_i(FIB_\text{p}, N, n_\infty)f_c(COL, c_\infty, 4)- \nonumber 
	d_f f_i(FIB_\text{p}, N, n_\infty)-\\
	& ~~~~~k_{\text{fnp}} f_H(N,x_{\text{fn}})FIB_\text{p}\label{2.10.1}, \\ 	
	\frac{dFIB_\text{m}}{dt} &= -\mu_{\text{fib}} FIB_\text{m}+ p_{\text{fibm}} f_i(FIB_\text{m}, N, n_\infty)f_c(COL, c_\infty, 4)-  d_f f_i(FIB_\text{m}, N, n_\infty) \nonumber \\ & 
	~~~~+ d_f f_i(FIB_\text{p}, N, n_\infty)-k_{\text{fnm}} f_H(N,x_{\text{fn}})FIB_\text{m} \label{2.10.2},\\ 	 
	\frac{dFIB_\text{a}}{dt} &= -\mu_{\text{fib}} FIB_\text{a}+ p_{\text{fiba}} f_i(FIB_\text{a}, N, n_\infty)f_c(COL, c_\infty, 4)+d_f f_i(FIB_\text{m}, N, n_\infty)- \nonumber \\ &~~~~~k_{\text{fna}}f_H(N,x_{\text{fn}})FIB_\text{a} + p_{\text{bl}}, \label{2.10.3}\\	 
	\frac{dCOL}{dt} &= k_{\text{cf}} f_H(FIB_a,x_\text{cf})f_c (COL, c_{f\infty}, h_{\text{fc}})- k_{\text{cn}}NCOLe^{-(T-T_0)^2}- \nonumber \\ &
	~~~k_{\text{cfr}}f_H(FIB_a,x_{\text{cfr}})s_h(N_{\text{crit}}-N)(1-f_c(COL, c_{fr\infty}, 12)), \label{2.10.4}\\	 
	\frac{dN}{dt}&= \frac{s_{\text{nr}}R_i}{\mu_{\text{nr}}+R_i}-k_{\text{nf}}FIB_aNe^{-(T-T_0)^2}-\mu_n N, \label{2.10.5}\\	 
	\frac{dP}{dt} &= k_{\text{pg}}(O_2)P \bigg( 1-\frac{P}{P_\infty}\bigg) -\frac{k_\text{pm} s_\text{m}P}{\mu_{\text{m}}+k_{\text{mp}} P}-k_{\text{pn}}P f_i(N, FIB_a, F_\infty) (1-g(O_2))\label{2.10.6}, 
\end{align} 
where
\begin{align*}
	f_H(x,V) &=\frac{x}{V+x},\\
	f_i(x,V,V_\infty) &= \frac{x}{1+(V/V_\infty)^2},\\
	f_c(COL,c_\infty,h_\infty) &= \frac{1}{1+(COL/c_\infty)^{h_\infty}},\\
	WS &=\text{max} \bigg( (1-COL) WS_0 \frac{1}{1-w_{\text{sg}} g(O_2)}, 0\bigg),\\
	s_h(x) &= \frac{1}{1+e^{-50x}},\\
	R_i &= f_i(k_{np}P+k_{nn} N^3+k_{nw}WS, FIB, F_\infty),\\
	g(V)&= 1-\frac{1.1}{1+ 0.1e^{-0.3(V-O_{\text{crit}})}},
\end{align*}
subject to  $FIB_p=10, P=0.5\;$ and $FIB_m=FIB_a = COL = N =0.$\\

\section{Positivity of solutions}

\begin{theorem}\label{se}
	Let $\displaystyle FIB_p(0)\geq 0, FIB_m(0)\geq 0, FIB_a(0)\geq 0, COL(0)\geq 0, N(0)\geq 0, P(0)\geq 0.$ The solutions $\displaystyle FIB_p(t), FIB_m(t), FIB_a(t), COL(t), N(t), P(t)$ are always positive for all $t\geq 0.$
\end{theorem}

\begin{proof}
	To prove Theorem \ref{se}, we show that the state variables/solutions of the system of equations (\ref{2.10.1})- (\ref{2.10.6}) are increasing functions in the interval $[t_0,t_f]$ and subject to the initial conditions $\displaystyle FIB_p(0)\geq 0, FIB_m(0)\geq 0, FIB_a(0)\geq 0, COL(0)\geq 0, N(0)\geq 0, P(0)\geq 0.$ In this case, we proceed by assuming that any of the variables will change its sign in the interval $t=\left(t_0,t_f\right),$ where $t_0$ is the initial time and $t_f$ is the final time, then we use proof by contradiction.\\
	
	Suppose any of the state variables are negative for all time $t\in (t_0,t_f),$ then there should exist a time point $\tau\in (t_0,t_f)$ where the particular variable changes its sign or hits the value zero for the first time, which we can write as $$\tau = \inf\{t: N(\tau)=0,\;\text{or}\; FIB_m(\tau)=0,\;\text{or}\; FIB_a(\tau)= 0,\;\text{or}\; COL(\tau)=0,\;\text{or}\; FIB_p(\tau)= 0,\;\text{or}\; P(t)= 0\}.$$
	
	Suppose $N$ is the first variable to hit zero, i.e. $N(\tau)=0,$ then for all $t\in (t_0,\tau),$ $N(t)>0,\; FIB_m(t)>0,\; FIB_a(t)>0,\; COL(t)>0,\; FIB_p(t)>0,\; P(t)>0,$ while $\displaystyle \frac{dN(t)}{dt}<0.$
	
	From Equation (\ref{2.10.5}), we observe that $$\frac{dN(t)}{dt}= \frac{k_{np}P(t)+k_{nn}N^3(\tau)+k_{nw}WS}{1+\left[\frac{FIB_a(t)}{F_{\infty}}\right]^2}-k_{nf}FIB_a(t)N(\tau)e^{-(T-T_0)^2}-\mu_nN(\tau)
	=\frac{k_{np}P(t)+k_{nw}WS}{1+\left[\frac{FIB_a(t)}{F_{\infty}}\right]^2}>0.$$
	
	This contradicts the assumption that $\frac{dN(\tau)}{dt}<0$ in the interval $t\in(t_0,\tau)$ and implying $N(t)$ cannot hit the value zero first. Thus $N(t)> 0$ for all $t\geq 0.$ Now suppose the variable $FIB_m$ is not positive for all time $t\geq 0,$ then there should exist a time point $\tau\in (t_0,t_f)$ where $FIB_m$ changes its sign or hits the value zero for the first time, which we can write as $$\tau = \inf\{t: FIB_m(\tau)=0,\;\text{or}\; FIB_p(\tau)=0,\;\text{or}\; FIB_a(\tau)= 0,\;\text{or}\; COL(\tau)=0,\;\text{or}\; P(\tau)= 0\}.$$
	
	For all $t\in (t_0,\tau),$ $FIB_m(t)>0, FIB_p(t)>0, FIB_a(t)>0, COL(t)>0, N(t)>0, P(t)>0,$ while $\displaystyle \frac{dFIB_m(t)}{dt}<0.$ From Equation (\ref{2.10.2}), we observe that\\
	
	$\displaystyle \frac{dFIB_m(t)}{dt}= -\mu_{fib}FIB_m(\tau)+\frac{p_{fibm}FIB_m(\tau)}{\left[1+(N(t)+N_{\infty})^2\right]\left[1+\left(\frac{COL(t)}{C_{\infty}}\right)^4\right]}-\frac{d_fFIB_m(\tau)}{1+(N+N_{\infty})^2}+ \frac{d_fFIB_p(t)}{1+(N(t)+N_{\infty})^2}  
	-\frac{k_{fnm}N(t)FIB_m(\tau)}{x_{fn}+N(t)}=\frac{d_fFIB_p(t)}{1+(N(t)+N_{\infty})^2}>0.$
	
	This contradicts the assumption that $\displaystyle\frac{dFIB_m(\tau)}{dt}<0,$ i.e. $\displaystyle FIB_m$ is decreasing in the interval $\displaystyle t\in(t_0,\tau)$ implying that $\displaystyle FIB_m(t)$ cannot hit the value zero first. Thus $\displaystyle FIB_m(t)> 0$ for all $\displaystyle t\geq 0.$ Similarly, we can show that neither of the state variables will hit the value zero in the interval $\displaystyle t=(t_0,t_f),$ implying that the state variables will remain positive for all time. 
\end{proof}
\section{Existence of equilibrium points}\label{sec}

In this section, we show the existence of steady states of system (\ref{2.10.1})-(\ref{2.10.6}). At steady states, 

$$\frac{dFIB_p}{dt}=\frac{dFIB_m}{dt}=\frac{dFIB_a}{dt}=\frac{dCOL}{dt}=\frac{dN}{dt}=\frac{dP}{dt}=0.$$
From equation (\ref{2.10.1}), we observe that $FIB_p=0$ or $$-\mu_{fib}+\frac{p_{fibp}}{[1+(N+N_\infty)^2]\left[1+\left(\frac{COL}{c_\infty}\right)^4\right]}-\frac{d_f}{1+(N+N_\infty)^2}-\frac{k_{fnp}N}{x_{fn}+N}=0.$$

From equation(\ref{2.10.2}), we observe that $FIB_m=0$ or $$-\mu_{fib}+\frac{p_{fibm}}{[1+(N+N_\infty)^2]\left[1+\left(\frac{COL}{c_\infty}\right)^4\right]}-\frac{d_f}{1+(N+N_\infty)^2}-\frac{k_{fnm}N}{x_{fn}+N}=0.$$

From Equation (\ref{2.10.6}), we observe that $P=0$ or $$o_{xy}\left[1-\frac{P}{P_\infty}\right]-\frac{k_{pm}s_{m}}{\mu_m+k_{mp}}-\frac{1.1k_{pn}N}{\left[1+0.1e^{-0.3[O_2-O_{crit}]}\right]\left[1+\left(\frac{FIB_a}{F_\infty}\right)^2\right]}=0.$$

When $\displaystyle FIB_p=FIB_m=P=0,$ then $$R_i=\frac{k_{np}P+k_{nn}N^3+k_{nw}WS}  {\left[1+\left(\frac{FIBp+FIBm+FIBa}{F_\infty}\right)^2\right]}$$

becomes $\displaystyle R_i=\frac{k_{nn}N^3+k_{nw}WS}{\left[1+\left(\frac{FIBa}{F_\infty}\right)^2\right]}.$ When the wound is healed, then the wound size $(WS)=0.$ This implies that
\begin{equation}R_i=\frac{k_{nn}N^3}{\left[1+\left(\frac{FIBa}{F_\infty}\right)^2\right]}=\frac{k_{nn}N^3F_\infty^2}{F_\infty^2+FIB_a^2}.\end{equation}

With $R_i$ so defined, and from Equation (\ref{2.10.5}), we obtain

$$s_{nr}k_{nn}N^3F_\infty^2-\left[\mu_{nr}(F_\infty^2+FIB_a^2)+k_{nn}N^3+F^2_\infty\right]\left[k_{nf}FIB_aNe^{-(T-T_0)^2}+\mu_nN\right]=0,$$\\

which gives $N=0$ or $s_{nr}k_{nn}N^2F_\infty^2-\left[\mu_{nr}(F_\infty^2+FIB_a^2)+k_{nn}N^3+F^2_\infty\right]\left[k_{nf}FIB_ae^{-(T-T_0)^2}+\mu_n\right]=0.$ If $N=0,$ Equation (\ref{2.10.3}) gives

\begin{equation}\label{gu}
	FIB_a[\mu_{fib}(1+n_\infty^2)(c_\infty^4+COL^4)-p_{fiba}c_\infty^4]=p_{bl}(1+n_\infty^2)(c_\infty^4+COL^4).
\end{equation} 

Equation (\ref{2.10.4}) can be written as 

\begin{align*}
	&\frac{k_{cf}FIB_a}{\left[x_{cf}+FIB_a\right]\left[1+\left(\frac{COL}{c_{f\infty}}\right)^{h_{fc}}\right]}-\frac{k_{cfr}FIB_a}{[x_{cfr}+FIB_a]\left[1+e^{-50(N_{crit}-N)}\right]}  \left[1-\frac{1.1}{1+\left(\frac{COL}{c_{fr\infty}}\right)^{12}}\right]\\   &-k_{cn}NCOLe^{-(T-T_0)^2}=0,\end{align*}

and when $N=0,$ it yields $FIB_a=0$ or 

\begin{equation}\label{gu1}
	\frac{k_{cf}c_{f\infty}^{h_{fc}}}{\left[c_{f\infty}^{h_{fc}}+COL^{h_{fc}}\right]}=\frac{k_{cfr}\left[COL^{12}-0.1C_{fr\infty}^{12}\right]}{\left[c_{fr\infty}^{12}+COL^{12}\right]\left[1+e^{-50N_{crit}}\right]}.
\end{equation}

$FIB_a=0$ implies $COL=0,$ which yields a trivial equilibrium point that we denote $$\mathcal{E}_0=\left[FIB_p^*,FIB_m^*,FIB_a^*,COL^*,N^*,P^*\right]=[0,0,0,0,0,0].$$

Eqn (\ref{gu}) gives the polynomial
\begin{equation}\label{gu2}
	COL^{12+h_{fc}}+\mathcal{A}_1COL^{12}-\mathcal{A}_2COL^{h_{fc}}-\mathcal{A}_3=0,
\end{equation}

where $\displaystyle\mathcal{A}_1=1-\left[1-\left(1+e^{-50N_{crit}}\right)k_{cf}c_{f\infty}^{h_{fc}}\right],$ $\displaystyle\mathcal{A}_2=0.1C_{fr\infty}^{12}$ and
\[\displaystyle\mathcal{A}_3=\frac{\left[1+e^{-50N_{crit}}\right]k_{cf}c_{f\infty}^{h_{fc}}c_{fr\infty}^{12}+0.1c_{f\infty}^{h_{fc}}c_{fr\infty}^{12}}{k_{cfr}}.\]

If we set $h_{fc}=1$ as articulated by \cite{segal2012differential}, then Equation (\ref{gu2}) becomes 

\begin{equation}\label{gu3}
	COL^{13}+\mathcal{A}_1COL^{12}-\mathcal{A}_2COL-\mathcal{A}_3=0.
\end{equation}

We observe that Equation (\ref{gu3}) has real coefficients and has only one sign change. By the Descartes rule of signs \citep{Descartes2017}, Equation (\ref{gu3}) has at least 1 positive real root, which we can denote as $COL^{**}.$ Now solving Equation (\ref{gu}), we get

\begin{align*}
	&FIB^{**}_a\left[\mu_{fib}\left(c_\infty^4+COL^{**4}+n^2_\infty c_\infty ^4+n^2_\infty COL^{**4}\right)-P_{fiba}c^4_\infty\right]\\ & = 
	p_{bl}\left[c_\infty^4+COL^{**4}+n^2_\infty c_\infty ^4+n^2_\infty COL^{**4}\right],   
\end{align*}

implying that

$$FIB_a^{**}= \frac{p_{bl}(1+n^2_{\infty})(c_\infty^4+COL^{**4})}{P_{fiba}c_\infty^4\left[\frac{\mu_{fib}(1+n_\infty^2)(c_\infty^{**4}+COL^{**4})}{P_{fiba} c_\infty^4}-1\right]},$$

simplifying to

$$FIB_a^{**}=\frac{p_{bl}(1+n_\infty^2)(c_\infty^4+COL^{**4})}{\mu_{fib}(1+n_\infty^2)(c_\infty^4+COL^{**4})\left[1-\frac{p_{fiba}c_\infty^4}{\mu_{fib}(1+n_\infty^2)(c_\infty^4+COL^{**4})}\right]}=\frac{p_{bl}}{\mu_{fib}[1-R_W]},$$

which is positive when $R_w<1,$ where \begin{equation}\label{rw}R_w=\frac{p_{fiba}c_\infty^4}{\mu_{fib}(1+n_\infty^2)(c_\infty^4+COL^{**4})}.\end{equation}

Thus the wound-free equilibrium point is given by 

$$\mathcal{E}_1=\left[FIB_p^{**},FIB_m^{**},FIB_a^{**},COL^{**},N^{**},P^{**}\right]=\left[0,0,\frac{p_{bl}}{\mu_{fib}[1-R_w]},COL^{**},0,0\right].$$

We observe that only activated fibroblast cells and collagen exist at the wound-free equilibrium point, as the body needs them to support normal tissue structure. Equation (\ref{rw}) also suggests a large collagen activity, a large fibroblast death rate ($\mu_{fib}$) and large inhibition constants, $n_\infty ~\text{and} ~c_\infty$ at the wound-free equilibrium. It is necessary to have a large $\mu_{fib}$ and a large $c_\infty$ for a large collagen activity to achieve normal collagen activity and avoid excessive deposition of collagen, a condition called scleroderma.

\section{Sensitivity analysis}

To gain insights into factors that influence the activity of fibroblast cells throughout the simulation period, we perform a global sensitivity/uncertainty analysis. Sensitivity analysis helps us to understand the effect of the complete parameter space on the output variable, in this case, the population density of fibroblast cells. The parameters are assumed to be uniformly distributed and are sampled using the Latin Hypercube Sampling (LHS) method, which is an efficient stratified Monte Carlo sampling that allows for simultaneous sampling of the multi-dimensional parameter space as fully outlined by \cite{Blower1994} and \cite{Hoare2008}. The simulation is carried out for a period of 30 days, with 1000 simulations per run. Partial rank correlation coefficients (PRCCs) are computed for each selected input parameter and the output variable during the simulation period. The magnitude and sign of the PRCC for each parameter are important as they determine the contribution of the parameter to the activity of fibroblast cells. The larger the positive PRCC value is or the smaller the negative PRCC value is, the more influence the parameter has on the activity of fibroblast cells \citep{Taylor1990}. Table \ref{tt1} shows the PRCC values for each parameter of our model, which indicates the correlation between the model parameters, and the population of fibroblast cells. PRCCs are considered significant if their $p-$values are less than $0.05.$ From Figure \ref{sy}, we observe that for proliferating fibroblast cells, the contact inhibition constant of FIBp reproduction by collagen $(C_{\infty}),$ is the most important factor as it has the highest contribution to the activity of proliferating fibroblast cells. The rate of production of collagen by active fibroblast cells $(k_{cf})$ and the rate of destruction of inflammatory cells by pathogens also contribute to an increase in the production of fibroblast cells, indicating a high level of uncertainty when these parameter values are changed. On the other hand, the fibroblast decay rate $(\mu_{fib})$ is highly negatively correlated to the activity of proliferating fibroblast cells, indicating that an increase in the values of this parameter results in a decrease in the relative activity of proliferating fibroblast cells. Pathogen growth under normal oxygen levels also contributes to a decrease in the activity of proliferating fibroblast cells, followed by the inhibition constant for contact inhibition in collagen degradation. Figure \ref{syfibm} shows the rate of production of collagen by active fibroblast cells, contact inhibition constant of active migrating fibroblast cells by collagen and rate of destruction of inflammatory cells by pathogens to have large positive PRCCs and decay rate of migrating fibroblast cells, inhibition exponent for contact inhibition in collagen production and inhibition constant of migrating fibroblast cells by inflammatory cells as factors with least PRCCs. A similar effect of these factors can also be observed in the activity of active fibroblast cells from Figure \ref{syfiba}. In the manufacture or administering of drug interventions, interventions that stimulate or increase the levels of parameters with positive PRCCs (sky-blue bars in Figures \ref{sy}, \ref{syfibm}, \ref{syfiba}) should be prioritized based on the order of importance. Similarly, factors with negative PRCCs (chocolate bars in Figures, \ref{sy}, \ref{syfibm}, \ref{syfiba}) contribute to slower healing of a wound. In section (\ref{sec}), we derived an important threshold denoted $(R_w)$ defined by Equation (\ref{rw}). We observed that when $R_w<1,$  the wound-free equilibrium point exists and is always positive. In order to maintain a wound-free equilibrium, we carry out a sensitivity analysis on $R_w$ and we identify important factors that influence this threshold. This analysis is important as it provides information that can be used in drug manufacture and administration. From Figure \ref{Fs}, we observe that the "contact inhibition constant of FIBx reproduction by collagen" ranks as the factor with the highest positive PRCC. If this factor is increased or continuously stimulated, then the wound-healing process is impaired. Similarly, the "reproductive rate of Fibroblast cells of type a" is positively correlated with $R_w,$ meaning the continued production of fibroblast cells of type a will increase $R_w$ in the process leading to prolonged healing of the wound. On the other hand, we observe that the "density of collagen at the wound-free equilibrium" has the highest negative PRCC. When this factor is increased, the value of $R_w$ is reduced, leading to an expedited healing process of the wound. In the design and manufacturing of drugs, if active ingredients accelerate the stimulation of collagen, then the wound-healing process is improved.

\begin{figure}[]
	\centering
	\includegraphics[width=1.0\linewidth]{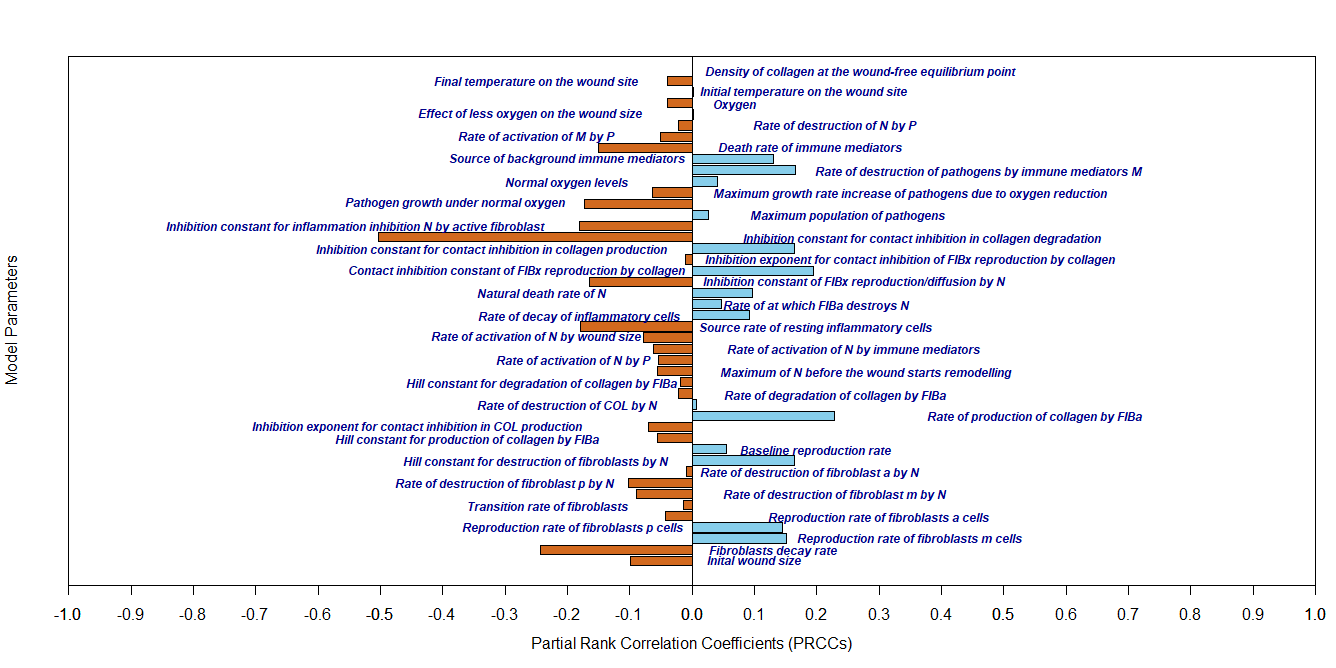}
	\caption{\begin{small}
			Tornado plot showing the partial rank correlation coefficients (PRCCs) of the input variables (model parameters), with respect to the proliferating fibroblast cells $(FIB_p).$ PRCCs are sampled using the Latin Hypercube Sampling (LHS) method, with 1000 simulations per run \citep{Blower1994}. 
	\end{small}}
	\label{sy}
\end{figure}

\begin{figure}[]
	\centering
	\includegraphics[width=1.0\linewidth]{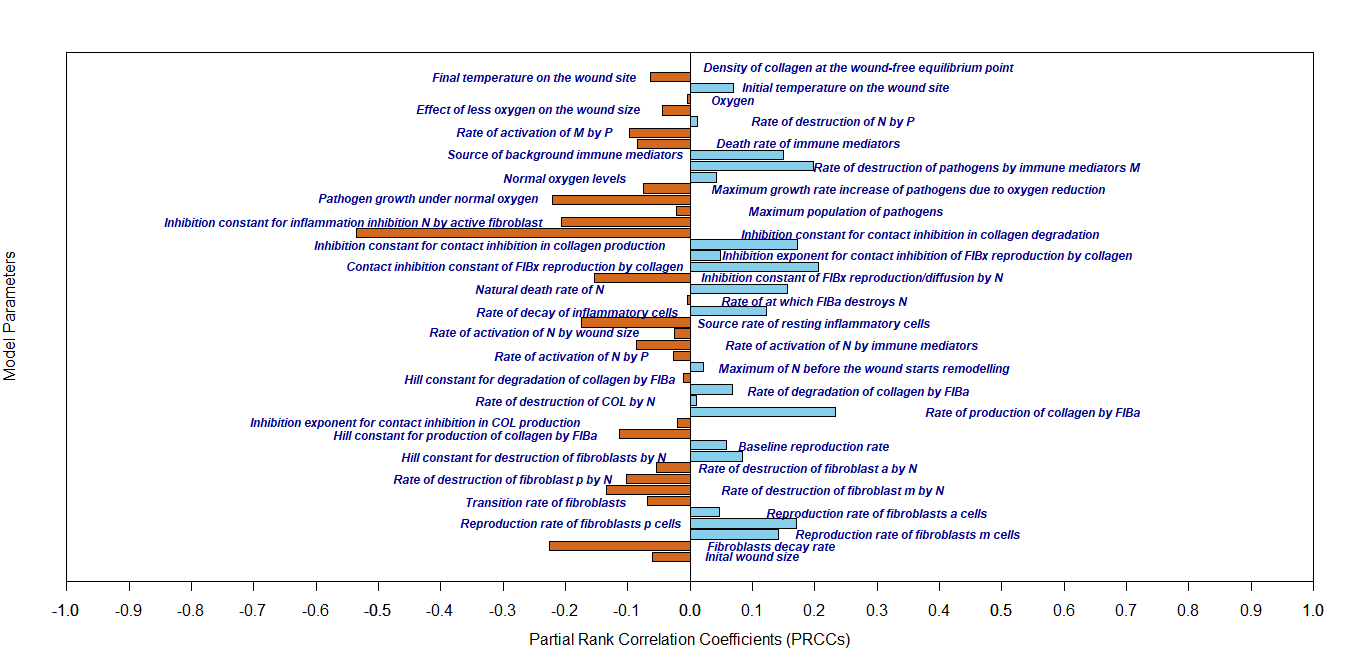}
	\caption{\begin{small}
			Tornado plot showing the partial rank correlation coefficients (PRCCs) of the input variables (model parameters), with respect to the migrating fibroblast cells $(FIB_m).$ PRCCs are sampled using the Latin Hypercube Sampling (LHS) method, with 1000 simulations per run \citep{Blower1994}
	\end{small}}
	\label{syfibm}
\end{figure}

\begin{figure}[]
	\centering
	\includegraphics[width=1.0\linewidth]{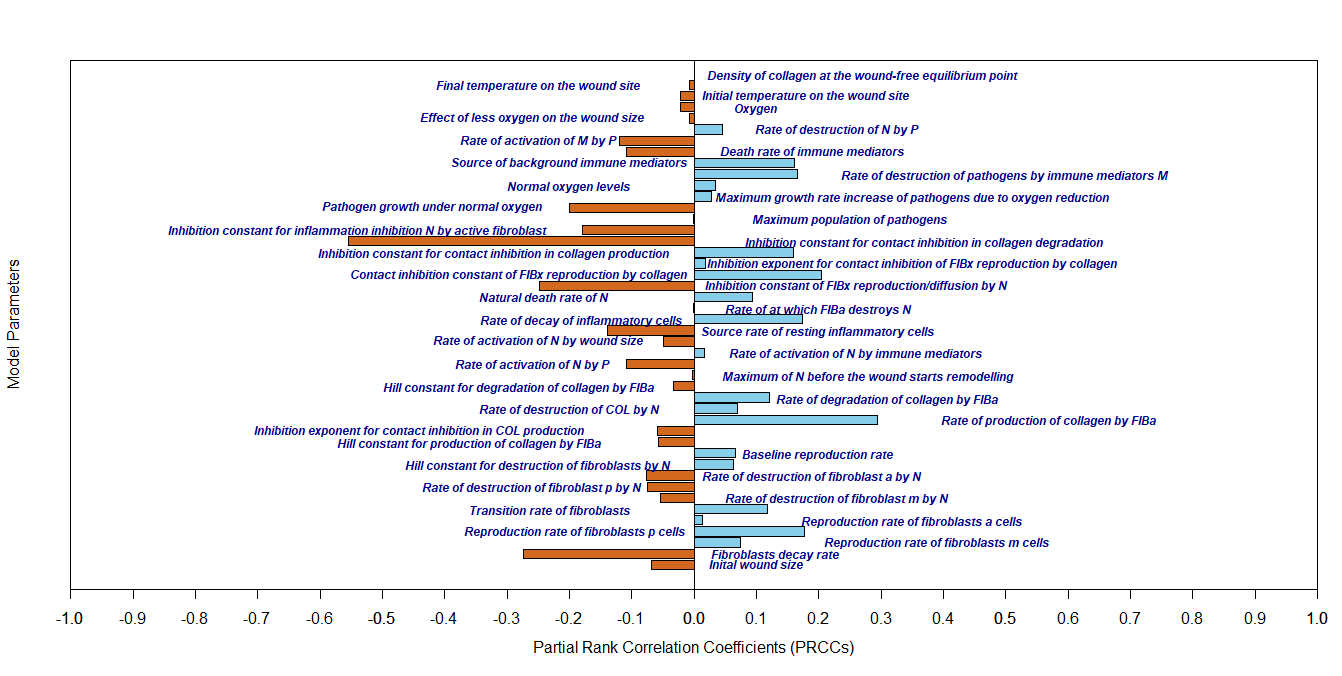}
	\caption{\begin{small}
			Tornado plot showing the partial rank correlation coefficients (PRCCs) of the input variables (model parameters), with respect to the active fibroblast cells $(FIB_a).$ PRCCs are sampled using the Latin Hypercube Sampling (LHS) method, with 1000 simulations per run \citep{Blower1994}
	\end{small}}
	\label{syfiba}
\end{figure}

\begin{table}[httb]
	\renewcommand{\arraystretch}{2.9}
	\small
	\centering
	\caption{Partial rank correlation coefficients for the sensitivity analysis.}
	\begin{adjustbox}{width=0.75\textwidth}
		\begin{tabular}{|c@{\hspace{1cm}}|c@{\hspace{1cm}}c@{\hspace{1cm}}c@{\hspace{1cm}}c@{\hspace{1cm}}|c@{\hspace{1cm}}c@{\hspace{1cm}}c@{\hspace{1cm}}c@{}|}
			\hline
			& &PRCCs&&&&$P-$ Values&&\\
			\hline
			Parameter& FIBp & FIBm&FIBa&$R_w$&FIBp & FIBm&FIBa&$R_w$\\
			\hline
			$WS_0$ & -0.09867 & -0.06092 & -0.06817 &0.07765 & $1.947 \times 10^{-3}$ & $1.563\times 
			10^{-3}$ &$2.444\times 10^{-4}$ & $1.926\times 10^{-3}$ \\ 
			$\mu_{fib}$ &  -0.24388 & -0.22603 & -0.27434 & -0.75356&  $1.385\times 10^{-3}$&$3.557\times 10^{-4}$ & $1.077\times 10^{-3}$&$5.636 \times 10^{-4}$   \\ 
			$P_{fibm}$ & 0.15053 & 0.14230 &0.07476 &-0.03683 & $1.731\times 10^{-3}$ & $1.481\times 10^{-3}$ &$2.772\times 10^{-5}$&$1.781\times 10^{-3}$  \\ 
			$P_{fibp}$ & 0.14405 & 0.16996 &0.17700 &-0.02578 &$1.678\times 10^{-3}$  &$7.192\times 10^{-4}$ &$4.301\times 10^{-4}$& $1.226\times 10^{-3}$ \\ 
			$P_{fiba}$ & -0.04277 & 0.04771 &0.01399 & 0.74508 & $6.812\times 10^{-4}$ &$2.233\times 10^{-4}$&$9.419\times 10^{-4}$ & $3.795\times 10^{-4}$  \\ 
			$d_f$ & -0.01399 & -0.06810 & 0.11765 &0.02959&$ 4.870\times 10^{-4}$  &$5.063\times 10^{-4}$ &$1.547\times 10^{-3}$ & $8.358\times 10^{-4}$\\ 
			$k_{fnm}$ & -0.08931 & -0.13465 &-0.05400 & -0.04281 &$1.162\times 10^{-3}$ & $1.300\times 10{-3}$&$5.063\times 10^{-4}$&$5.709\times 10^{-4}$  \\ 
			$k_{fnp}$ & -0.10283 & -0.10277 & -0.0743 &0.05604& $2.576\times 10^{-4}$& $1.789\times 10^{-3}$&$1.426\times 10^{-3}$&$2.589\times 10^{-3}$  \\ 
			$k_{fna}$ & -0.00947  & -0.05397 &-0.07692&0.01939&$2.684\times 10^{-4}$  &$6.435\times 10^{-4}$ &$4.339\times 10^{-4}$ &$3.385\times 10^{-4}$  \\ 
			$x_{fn}$ & 0.16487 &  0.08310 & 0.06378 & -0.00865&$1.179\times 10^{-3}$ &$1.329 \times 10^{-3}$ &$2.429\times 10^{-4}$&$9.428\times 10^{-4}$ \\ 
			$P_{bl}$ & 0.05538 & 0.05751 &0.06629 &0.00545 &$5.658\times 10^{-4}$  & $7.873\times 10^{-4}$ &$1.405\times 10^{-3}$& $2.367\times 10^{-4}$ \\ 
			$x_{cf}$ & -0.05621 & -0.11381 &-0.05725&-0.01416 &$1.343\times 10^{-3}$ &$6.316\times 10^{-4}$&$1.428\times 10^{-3}$&$6.119\times 10^{-5}$  \\ 
			$h_{fc}$ & -0.06992 & -0.02031 &-0.05880 &-0.00240 & $1.305\times 10^{-3}$&$1.093\times 10^{-3}$ &$2.546\times 10^{-4}$ &$6.394\times 10^{-4}$  \\ 
			$k_{cf}$ & 0.22853 & 0.23347 & 0.29336& -0.02938 &  $7.535\times 10^{-4}$ &$1.413\times 10^{-3}$ &$ 2.025\times 10^{-3}$&$5.097\times 10^{-5}$  \\ 
			$k_{cn}$ & 0.00747 & 0.01064 & 0.06959  & -0.00355 &$9.616\times 10^{-4}$ &$1.050 \times 10^{-4}$ &$2.848\times 10^{-3}$ &$1.494\times 10^{-3}$ \\ 
			$k_{cfr}$ & -0.02162 & 0.06748 & 0.12020 &0.04583& $1.171\times 10^{-3}$ &$4.349\times 10^{-4}$ &$5.255\times 10^{-4}$ &$1.336\times 10^{-3}$  \\ 
			$x_{cfr}$ & -0.01818 & -0.01084 & -0.03382 & 0.00145& $3.110\times 10^{-4}$&$8.742\time 10^{-4}$ &$9.761\times 10^{-4}$ & $2.076\times 10^{-3}$\\ 
			$N_{crit}$ & -0.05514 & 0.02169 & -0.00306 &-0.00181 &$1.799\times 10^{-3}$&$8.259\times 10^{-5}$ &$ 1.226\times 10^{-3}$ &$2.219\times 10^{-4}$\\ 
			$k_{np}$ & -0.05348 & -0.02654 & -0.10857 &-0.03139 &$1.046\times 10^{-3}$  &$2.216\times 10^{-3}$ &$8.319\times 10^{-4}$ &$4.372\times 10^{-4}$\\ 
			$k_{nn}$ & -0.06270 & -0.08609 & 0.01685 & -0.03661 & $3.488\times 10^{-4}$& $1.426\times 10^{-3}$ &$4.652\times 10^{-4}$ & $5.024\times 10^{-4}$\\ 
			$k_{nw}$ & -0.07879 & -0.02510 & -0.04957 &-0.02993&$2.506\times 10^{-4}$ & $1.350\times 10^{-3}$ &$2.613\times 10^{-4}$ & $1.369\times 10^{-3}$  \\ 
			$s_{nr}$ & -0.17893 & -0.17430 & -0.13907 & 0.03006& $1.682\times 10^{-3}$&$4.601\times 10^{-4}$ &$1.339\times 10^{-3}$&$1.615\times 10^{-3}$ \\ 
			$\mu_{nr}$ & 0.09213 & 0.12277 & 0.17378 &0.04624  &$8.096\times 10^{-5}$  &$1.001\times 10^{-3}$&$1.851\times 10^{-3}$ & $1.436\times 10^{-3}$ \\ 
			$k_{nf}$ & 0.04648 & -0.00371 &-0.00100 &-0.02142 &$5.947\times 10^{-4}$ &$1.385\times 10^{-3}$ &$ 1.058\times 10^{-3}$ & $9.612\times 10^{-4}$\\ 
			$\mu_n$ & 0.09691 & 0.15627 & 0.09301  &-0.03560 &$1.505\times 10^{-3}$ &$9.161\times 10^{-4}$ &$2.618\times 10^{-5}$ &$6.924\times 10^{-4}$ \\ 
			$n_\infty$ & -0.16515 & -0.15408 &-0.24738 & -0.40913 &$4.570\times 10^{-4}$ &$6.897\times 10^{-4}$ &$8.571\times 10^{-4}$ &$5.695\times 10^{-4}$  \\ 
			$c_\infty$ & 0.19379 & 0.20600  &  0.20461 &0.97729  &$1.355\times 10^{-3}$  &$8.674\times 10^{-4}$ &$6.526\times 10^{-4}$& $6.797\times 10^{-5}$\\ 
			$h_c$ & -0.01071 & 0.04882 & 0.01862 &-0.02614 & $1.749\times 10^{-4}$  &$4.804\times 10^{-4}$ &$7.905\times 10^{-4}$& $1.645\times 10^{-4}$  \\ 
			$c_{f_\infty}$ & 0.16431 & 0.17143 & 0.15863 &-0.02076&$1.722\times 10^{-3}$ &$1.422\times 10^{-3}$ &$7.453\times 10^{-5}$ &$1.433\times 10^{-5}$ \\ 
			$c_{fr_\infty}$& -0.50249 & -0.53559 & -0.55395  &0.01476 &$8.677\times 10^{-4}$  &$2.007\times 10^{-4}$&$6.980\times 10^{-4}$& $7.806\times 10^{-5}$\\ 
			$c_{fr_\infty}$& -0.18005 & -0.20618 & -0.17916 & -0.03991 &$5.111\times 10^{-4}$  &$2.106\times 10^{-3}$&$1.322\times 10^{-3}$ & $1.794\times 10^{-4}$\\ 
			$F_\infty$ & 0.02681 & -0.02193 &-0.00168& -0.01756  &$2.947\times 10^{-4}$ &$1.947\times 10^{-4}$ &$1.796\times 10^{-3}$& $3.793\times 10^{-4}$ \\ 
			$P_\infty$ & -0.17221 & -0.22158 & -0.19981 &-0.02237 &$1.703\times 10^{-3}$ &$5.279\times 10^{-4}$&$1.834\times 10^{-4}$ & $5.908\times 10^{-5}$ \\ 
			$k_{pgo}$ & -0.06425 & -0.07569 & 0.02787 &-0.04016& $1.167\times 10^{-3}$ &$5.404\times 10^{-4}$ &$5.418\times 10^{-4}$&$2.580\times 10^{-4}$  \\ 
			$\beta_p$ & 0.04008 & 0.04181 & 0.03435 &0.01007 & $1.120\times 10^{-3}$ &$2.816\times 10^{-4}$ &$1.719\times 10^{-3}$&$8.404\times 10^{-4}$  \\ 
			$O_{crit}$ & 0.16621 & 0.19784 & 0.16629 &-0.01851 &$1.310\times 10^{-3}$ &$5.075\times 10^{-5}$ & $5.823\times 10^{-4}$& $6.344\times 10^{-4}$   \\ 
			$k_{pm}$ & 0.12979 & 0.14933 & 0.16146 &0.02416 &$1.804\times 10^{-3}$ &$2.558\times 10^{-5}$&$9.782\times 10^{-5}$ &$1.171\times 10^{-4}$ \\ 
			$s_m$ & -0.14993 &  -0.08384 & -0.10824 & -0.03805& $1.864\times 10^{-3}$ & $1.098\times 10^{-3}$&$3.651\times 10^{-4}$ & $4.047\times 10^{-4}$\\ 
			$\mu_M$ & -0.05016 & -0.09668 & -0.11939 &-0.04188&$9.971\times 10^{-4}$ & $1.634\times 10^{-4}$&$3.150\times 10^{-4}$ & $5.715\times 10^{-4}$\\ 
			$k_{mp}$ & -0.02162 & 0.01094 & 0.04617 &0.01713 &$3.076\times 10^{-4}$  &$3.636\times 10^{-3}$ & $2.467\times 10^{-4}$&$1.111\times 10^{-3}$ \\ 
			$k_{pn}$ & 0.00274 & -0.04414 & -0.00743 &-0.00713 &$1.570\times 10^{-3}$ &$7.871\times 10^{-4}$&$2.268\times 10^{-4}$ &$1.361\times 10^{-4}$  \\ 
			$O_2$ & -0.04025 & -0.00443 & -0.02163 & 0.01384& $1.355\times 10^{-3}$&$9.872\times 10^{-4}$  &$8.208\times 10^{-4}$ & $2.065\times 10^{-4}$ \\ 
			$T_O$ & 0.00188 & 0.06930 &-0.02238  & -0.04302&$4.878\times 10^{-4}$ &$5.122\times 10^{-4}$ &$1.186\times 10^{-3}$ &$6.406\times 10^{-4}$ \\ 
			$T$ & -0.03930 & -0.06356  & -0.00787 & 0.04653&$1.103\times 10^{-3}$ &$9.170\times 10^{-5}$ &$ 1.943\times 10^{-4}$ &  $1.937\times 10^{-3}$ \\ 
			$COL^{**}$&-&-&-&--0.89183& -& -&- &- $2.844\times 10^{-4}$ \\
			\hline
		\end{tabular}
	\end{adjustbox}
	\label{tt1}
\end{table}

\begin{figure}[]
	\centering
	\includegraphics[width=1.0\linewidth]{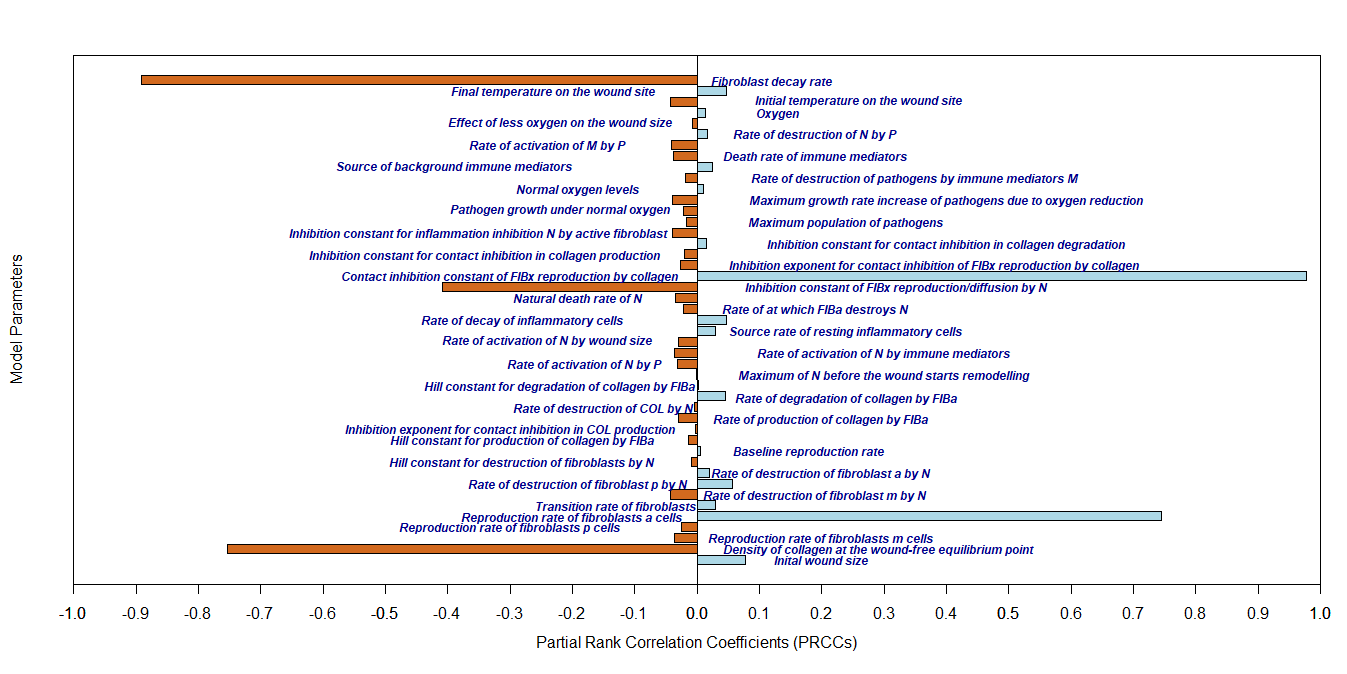}
	\caption{\begin{small}
			Tornado plot showing the partial rank correlation coefficients (PRCCs) of the input variables (model parameters), with respect to the threshold $(R_w).$ PRCCs are sampled using the Latin Hypercube Sampling (LHS) method, with 1000 simulations per run \citep{Blower1994}
	\end{small}}
	
	\label{Fs}
\end{figure}

\section{Numerical simulations}\label{rd}
In this section, we present and analyze the numerical simulations conducted on the model (\ref{2.10.1})-(\ref{2.10.6}). These simulations employ the TensorFlow library in Python, and the corresponding implementation can be found in the repository provided here (\cite{Alinafegithublink}). The purpose of these simulations is to enhance our comprehension of the dynamics within various cell activities and their influence on the wound healing process. The parameter values used in the numerical simulations are specified in Table \ref{long1}, unless noted otherwise.

\subsection{Oxygen on inflammatory cells and pathogen growth}
Oxygen plays a big role in the growth of pathogens. The growth rate of pathogens depends on the oxygenation of the tissue cells on the wound as modelled in equations (\ref{o2andpath}). Pathogen growth rate increases for oxygen levels less than critical $(25)$, and the growth rate remains constant $(0.55)$ for levels at and above critical level. 
\begin{figure}[]
	\centering
	\includegraphics[width=1.0\linewidth]{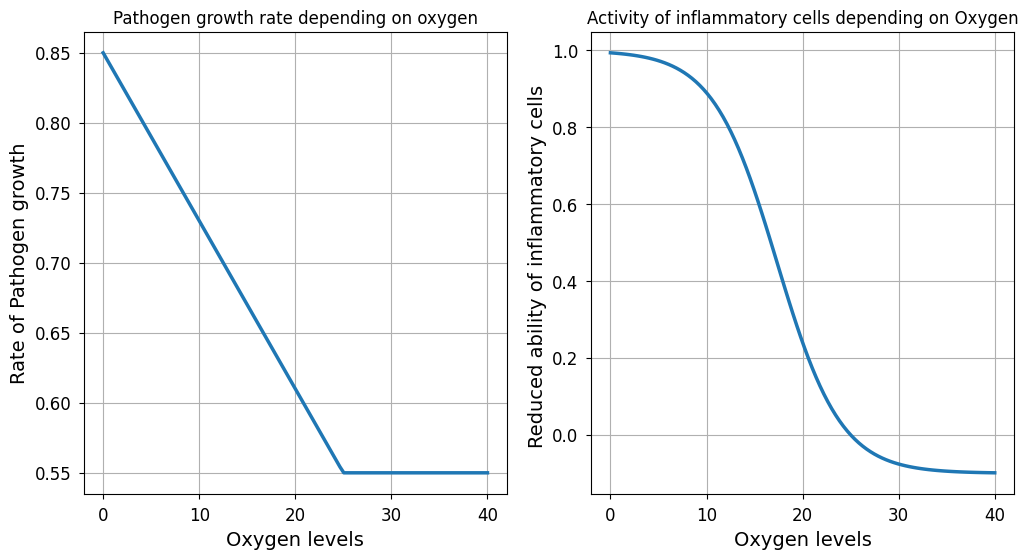}
	\caption{\begin{small}
			Pathogen growth rate under different oxygen levels (\textit{Left}) and the decreased ability of inflammatory cells under different oxygen levels (\textit{Right}).
	\end{small}}
	
	\label{Fig3.1}
\end{figure}
Oxygen also plays an important role in regulating the ability of inflammatory cells. Equation (\ref{oxy}) models the decreased abilities of inflammatory cells in different oxygen levels. Figure \ref{Fig3.1} (Right) indicates a decrease in the activity of inflammatory cells in high-oxygenated environments. This means increased and prolonged inflammation in less oxygenated wounds. On the other hand, we observe decreased activities of inflammatory cells in high-oxygenated wound tissues. This means an increased activity of inflammatory cells in well-oxygenated wounds. 

\subsection{Effect of inflammation on fibroblast cells}
Figure \ref{Fig3.2} explains the hyperinflammatory response to the activity of fibroblast cells. The activity of proliferating and active fibroblast cells decreases in the presence of inflammation and there will be an increased activity of the fibroblast cells in the absence of inflammation. This accurately reflects the biology of the wound healing process in which the activity of fibroblast cells is inhibited by the activity of inflammatory cells and so we have an increased activity of fibroblast cells in abscence of inflammatory response. 

\begin{figure}[]
	\centering
	\includegraphics[width=1.0\linewidth]{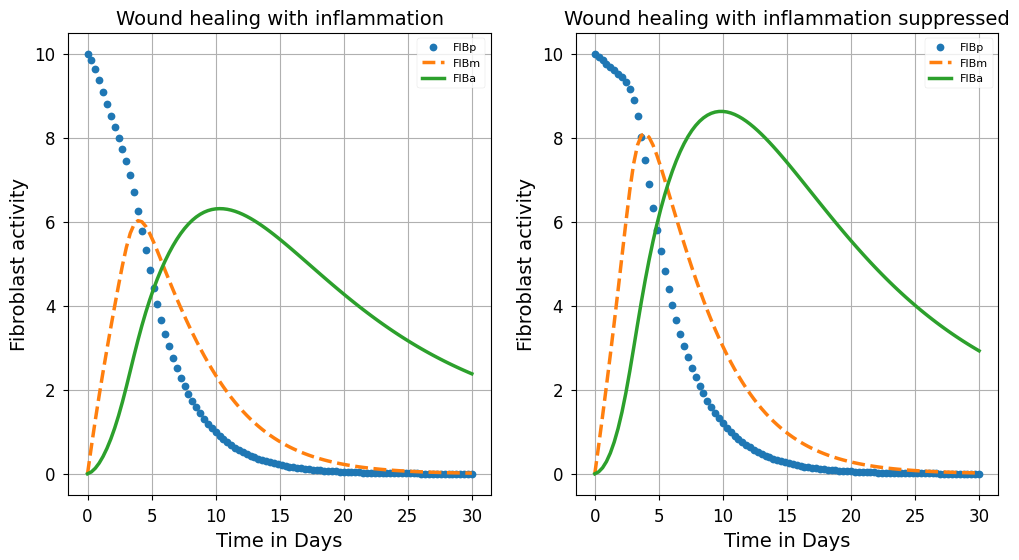}
	\caption{\begin{small}
			Fibroblast activity with inflammatory response (\textit{Left}) and without inflammatory response (\textit{Right}): Setting initial conditions $FIBp=10, P_0=0.5,~$and the other variables set to zero.
	\end{small}}
	
	\label{Fig3.2}
\end{figure}

\subsection{Collagen, inflammation, and pathogens in a normal healing wound}
Figure \ref{Fig3.3} presents collagen and pathogen activities in a normal wound healing process. There is a rise of collagen (Left) activity from day $0\;$ to around day $6\;$ from which there is a constant activity. 

The presence of an initial wound will activate the recruitment of pathogens on the wound, setting an initial wound size at $0.5,\;$ the body responds with inflammation to destroy the pathogens. Normally, pathogens will be cleared around day $8\;$ of the initial wound.
\begin{figure}[]
	\centering
	\includegraphics[width=1.0\linewidth]{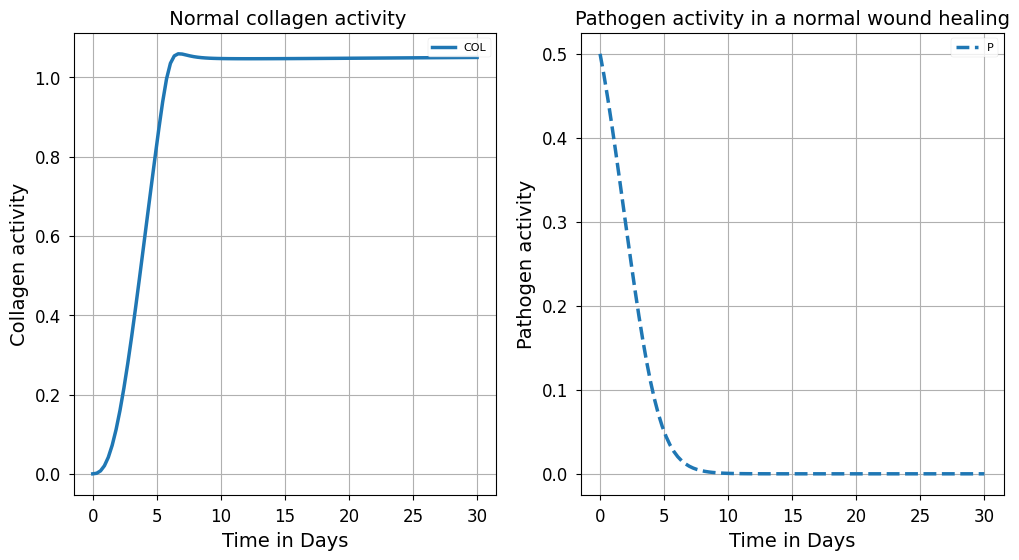}
	\caption{\begin{small}
			Collagen (\textit{Left}) and pathogen (\textit{Right}) activities during a normal wound healing process.
	\end{small}}
	
	\label{Fig3.3}
\end{figure} 
\begin{figure}[]
	\centering
	\includegraphics[width=1.0\linewidth]{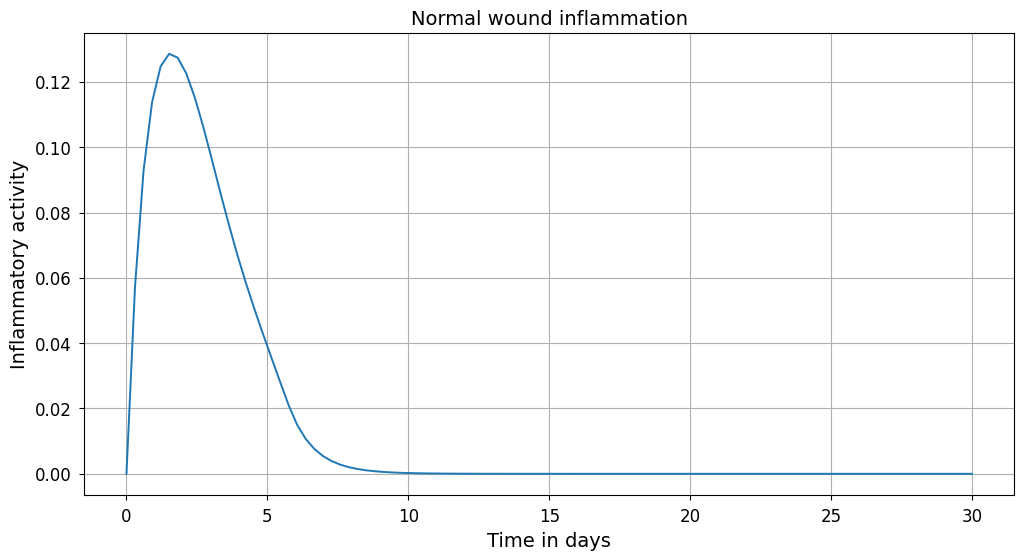}
	\caption {\begin{small}
			The activity of inflammatory cells in a healing wound.
	\end{small}}
	
	\label{Fig3.4}
\end{figure}The body responds to the pathogens on the wound by producing inflammatory cells. This response is activated a few hours after tissue damage. The absence of pathogens and damaged tissue cells means there is no activity of inflammatory cells. The activity of inflammatory cells is at maximum after around $3\;$ days of the wound after which it decreases to zero around day $8$.

\subsection{Effects of decreased oxygen on fibroblast cells}
The model was used to investigate the effect of oxygen reduction on the activity of fibroblast cells. Figure \ref{Fig3.5} compares fibroblast activities in well-oxygenated and oxygen-deprived environments. There is a clear indication of a reduced activity of fibroblast cells in less oxygenated tissue cells. In normal oxygenated wound surfaces, migrating and active fibroblast cells have a maximum relative activity of around $6.4~ (\text{around day 4})~ \text{and} ~7~ (\text{around day 10})~$ respectively and a maximum relative activity of around $5.1~ (\text{around day 4})~ \text{and} ~5.7~ (\text{around day 10})~$ in an oxygen reduced environment. There is also an observed decreased activity of migrating fibroblast cells within the first $5$ days.
\begin{figure}[]
	\centering
	\includegraphics[width=1.0\linewidth]{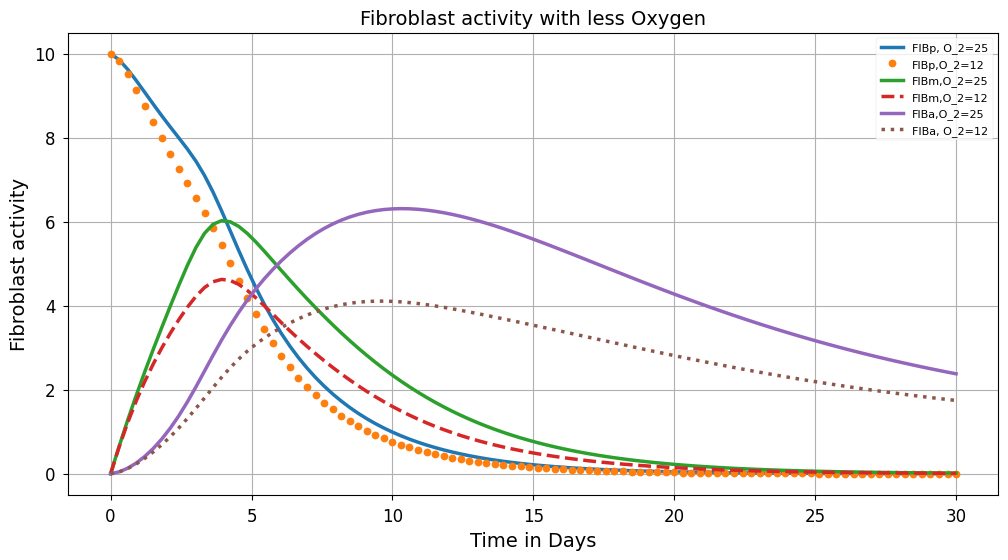}
	\caption{\begin{small}
			Effect of reduced oxygen on fibroblast activity: $O_2=12, FIB_p=10, P_0=0.5.~$ All the other variables ($FIB_m, FIB_a, COL, N$) were set to zero.
	\end{small}}
	
	\label{Fig3.5}
\end{figure}

\subsection{Effect of oxygen on collagen, pathogen and inflammatory  activities}
In an oxygen-depleted environment, the activities of collagen and pathogens (Figure \ref{Fig3.6}) are also affected. In less oxygenated tissues, there is a small decrease of collagen activity between the first $15$ days. However, oxygen greatly affects the growth of pathogens in the wound. As shown in Figure \ref{Fig3.1} (\textit{Left}), less oxygen levels increase the rate of pathogen growth and keep oxygen levels on or above normal levels (25), keeping the rate of growth of pathogens at constant (0.55). This pathogen growth dynamics directly correlates with the relative activity of pathogens on the wound. In normal oxygen levels, pathogen activity is cleared on around day $8~$ while in oxygen-depleted wound environments, there is active pathogen activity up to around day $19$. This can lead to prolonged wound recovery due to increased tissue damage and hence cause chronic wounds.
\begin{figure}[]
	\centering
	\includegraphics[width=1.0\linewidth]{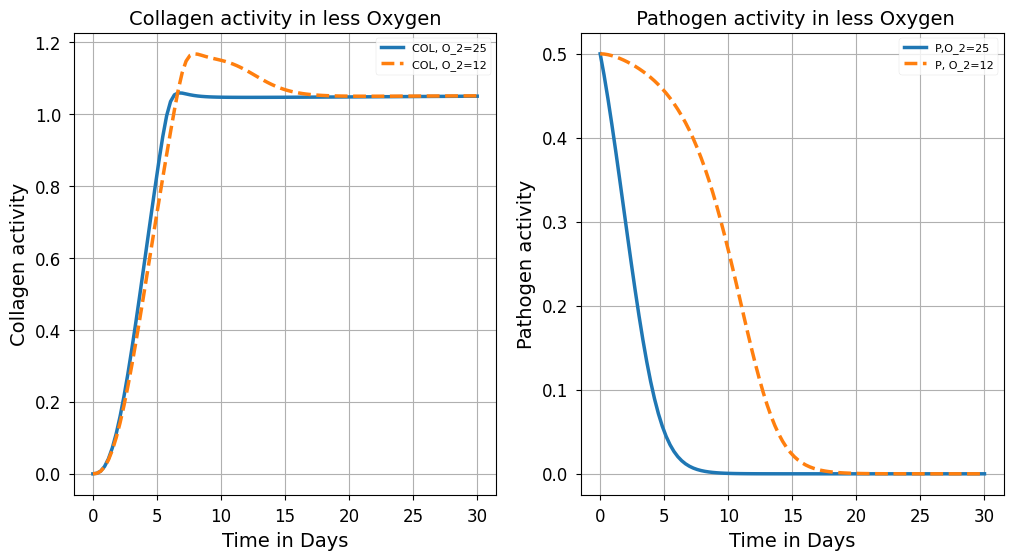}
	\caption{\begin{small}Response of collagen (\textit{Left}) and pathogen activity (\textit{Right}) to less oxygen: $O_2=12, FIB_p=10, P_0=0.5.~$ All the other variables ($FIB_m, FIB_a, COL, N$) were set to zero.
	\end{small}}
	
	\label{Fig3.6}
\end{figure}Figure \ref{Fig3.7} explains the activity of inflammatory cells in less oxygenated wounds. We observe an activity of inflammatory cells up to around day $17~$ and the activity is not smooth from around day $9$. This causes instability in the swelling of the wound. This prolonged and unstable inflammatory activity is also a driver of a chronic wound, thus, less oxygenated acute wounds may result in chronic wounds. Unlike under normal oxygenation levels where inflammatory activity is up to day $9~$.
\begin{figure}[]
	\centering
	\includegraphics[width=1.0\linewidth]{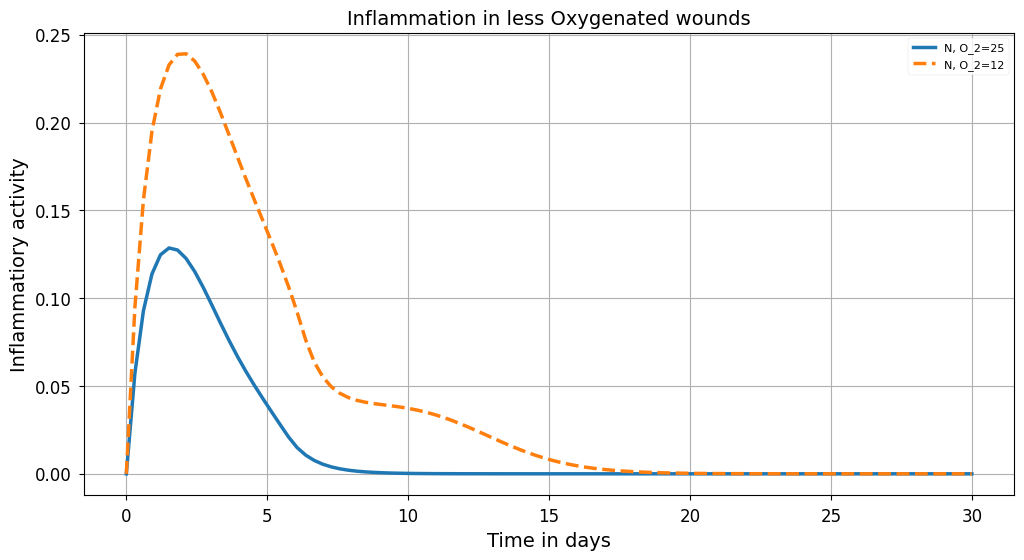}
	\caption{\begin{small}
			Response of activity of inflammatory cells to less oxygen: $O_2=12, FIB_p=10, P_0=0.5.~$ All the other variables ($FIB_m, FIB_a, COL, N$) were set to zero.
	\end{small}}
	
	\label{Fig3.7}
\end{figure} 

\subsection{Unclean wounds}
Figure \ref{Fig3.8} shows fibroblast activity where there are high levels of pathogen activity. In this case, we have maximum activity of migrating fibroblast cells at $6.2~$ and a maximum activity of around $6.3~$ for active fibroblast cells, from the normal maximum activities of around $6.4$ and $7$ for migrating and active fibroblast cells respectively. This means high levels of pathogens cause a small reduction of the activities of fibroblast cells. 

\begin{figure}[]
	\centering
	\includegraphics[width=1.0\linewidth]{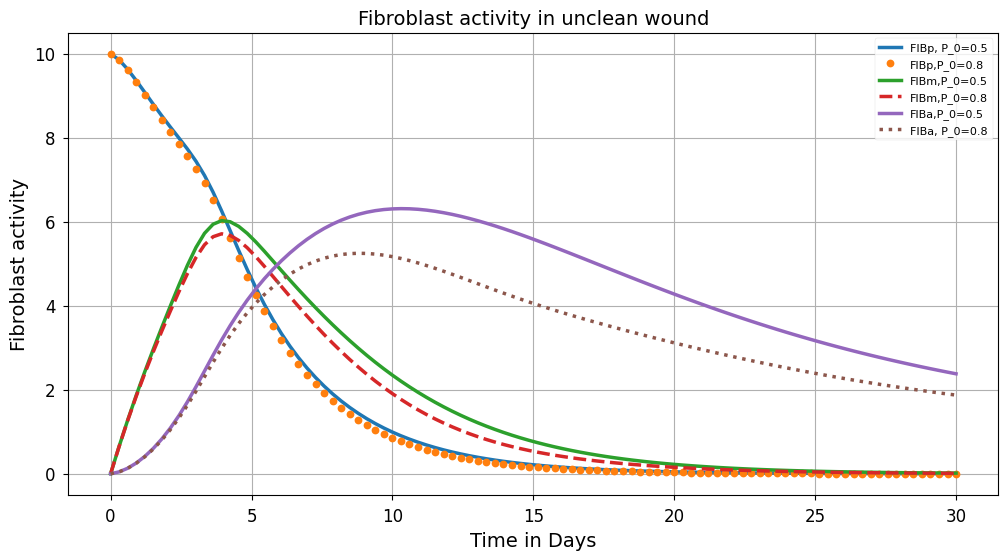}
	\caption{\begin{small}
			Activity of fibroblast cells in high levels of pathogen activity: $FIB_p=10, P_0=0.8.~$ All the other variables ($FIB_m, FIB_a, COL, N$) were set to zero. Oxygen levels are set at $25$.
	\end{small}}
	
	\label{Fig3.8}
\end{figure} 

The inhibition of active fibroblast cell activity explained by increased pathogen activity will also impact the activity of collagen as shown in Figure \ref{Fig3.9} (\textit{Left}) as compared to smooth collagen activity in clean wounds. 

\begin{figure}[]
	\centering
	\includegraphics[width=1.0\linewidth]{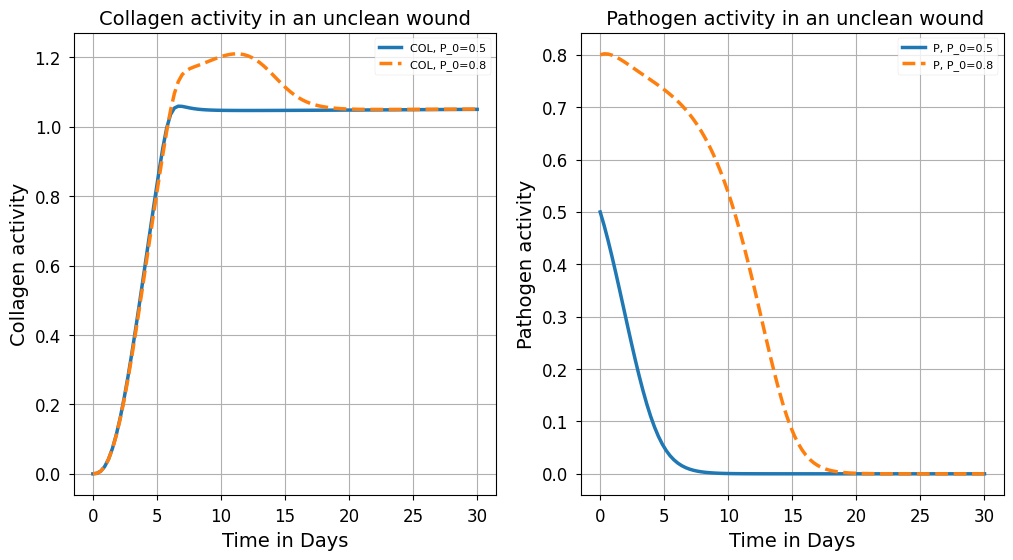}
	\caption{\begin{small}
			Collagen (\textit{Left}) and pathogen (\textit{Right}) activities in high pathogen activity levels: $FIB_p=10, P_0=0.8.~$ All the other variables ($FIB_m, FIB_a, COL, N$) were set to zero. Oxygen levels are set at $25$.
	\end{small}}
	
	\label{Fig3.9}
\end{figure}
Figure \ref{Fig3.9} (\textit{Right}) presents pathogen activity in unclean wounds. It is shown that high levels of pathogen activity  increases the time frame for pathogen activity cells to clear. In normal pathogen activity levels, it takes only around $9$ days to clear pathogen activity on the wound site unlike in the case of high pathogen activity levels, where it takes around $20$ days to clear pathogen activity. This is another indicator of delayed healing and a possible chronic wound.
\begin{figure}[]
	\centering
	\includegraphics[width=1.0\linewidth]{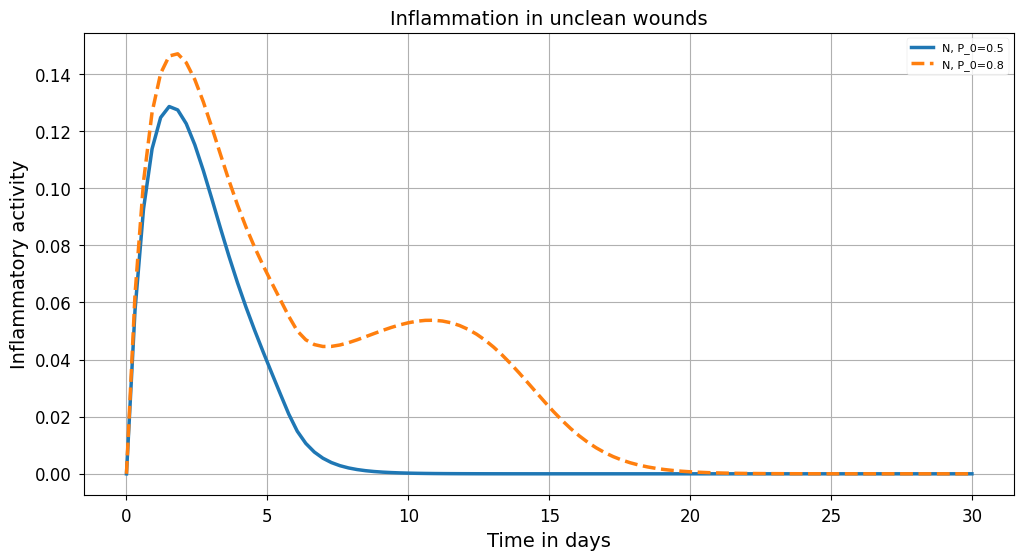}
	\caption{\begin{small}
			Activity of inflammatory cells under high levels of pathogen activity: $FIB_p=10, P=0.8.~$ All the other variables ($FIB_m, FIB_a, COL, N$) were set to zero. Oxygen levels set at $25$.
	\end{small}}
	
	\label{Fig3.10}
\end{figure}

The activity of inflammatory cells is also affected when there is increased pathogen activity. As shown in Figure \ref{Fig3.10} above, the activity of inflammatory cells is not smooth. This is a cause of irregular swelling which is also an indicator of a chronic wound. It is also shown that it takes around $20~$ days of activity of inflammatory cells while it takes around $8~$ days when the wound site is kept clean. 

\subsection{Acute and chronic wounds}
The model was used to understand cell activities in acute and chronic wounds. In Figure \ref{Fig3.11}, (a) and (b), shows cells activities in an acute wound (wound size set at $0.6$) and (c) and (d) explains cell activities in a chronic wound (represented by setting wound size at 0.8). There is observable increased inhibition of activities of migrating and active fibroblast cells in chronic wounds (c) than in acute wounds (a). There is also a reduction in the activity of collagen activity in chronic wounds (d) than in acute wounds (b). The activity of pathogens is the same assuming infection-free acute and chronic wounds. However, there is a greater increase in the activity of inflammatory cells in chronic wounds (d) than in acute wounds(b). Although there is normal pathogen activity, the increased inflammatory response may also be caused by other factors associated with chronic wounds such as diseases, for example, diabetes and vascular diseases, and stress which in turn cause an inbalance of cell activities on the wound site. This increased activity of inflammatory cells in (d) explains the increased inhibition of collagen activity which is a result of reduced recruitment of fibroblast cells, a major contributing factor of most chronic wounds, as shown in (c) which are a source of collagen protein.     
\begin{figure}[]
	\centering
	\includegraphics[width=1.0\linewidth]{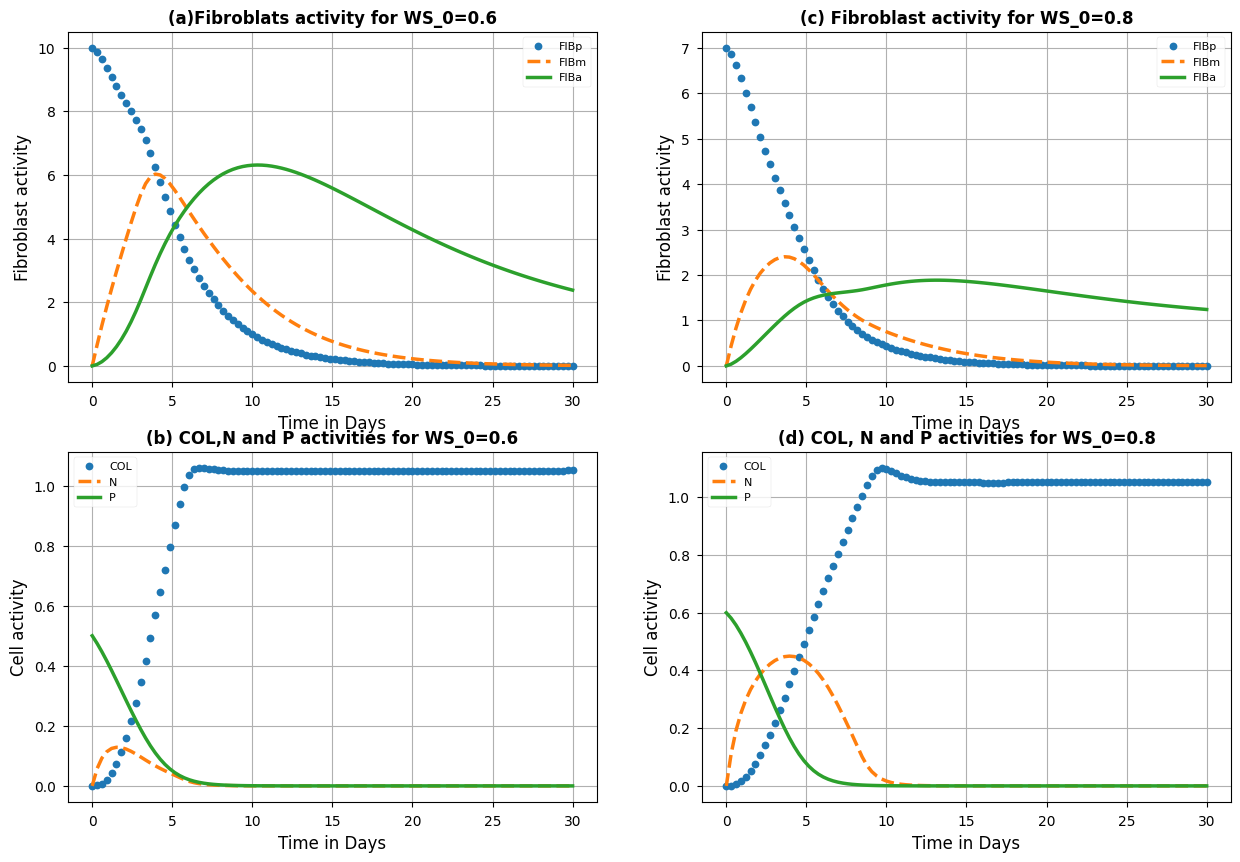}
	\caption{\begin{small}Healing behaviour of an acute wound (a) and (c) and a chronic  wound (b) and (d) assuming no increased pathogen infection in both cases. Setting  $WS_0=0.6$ for an acute wound and $WS=0.8,FIB_\text{p}=7.0, FIB_\text{m}=0, FIB_\text{a}=0, COL=0, N=0, P=0.6$ for chronic wounds.
	\end{small}}
	
	\label{Fig3.11}
\end{figure}

\subsection{Effect of systemic factors on wound healing process}
A study by  \cite{guo2010factors} on factors affecting wound healing classifies age, stress, diabetes, obesity, medication, sex hormones in aged individuals, alcohol consumption, smoking and nutrition as systemic factors that affect a wound healing process. This study presents effects of stress, age, alcohol consumption and smoking on wound healing process.

 \cite{swift2001age} finds that aged individuals have a delayed wound healing process due to reduced production of inflammatory cells. There are similar results on the study of effects of stress on wound healing process by  \cite{boyapati2007role}. Stress causes production of hormones which reduces the production of specific inflammatory cells and mediators.  \cite{guo2010factors} summarizes the effects of alcohol and smoking to the wound healing process. According to their study, alcohol consumption reduces resistance of inflammatory cells to pathogens and delays angiogenesis and chemical compounds in smoke like nicotine and carbon dioxide reduces tissue oxygenation. This reduces production of specific inflammatory cells and affect fibroblast migration. In general, ageing, stress, alcoholism and smoking reduces the production rate of resting inflammatory cells and proliferating fibroblast cells. These effects of the systemic factors are considered in the model by reducing the relative source of resting inflammatory cells $(s_{\text{nr}})$ and reducing initial proliferating fibroblast cells. Figure \ref{Fig3.12}, shows the activity of fibroblast cells with $s_{\text{nr}}= 1.2$ and $FIB_\text{p}=8$. 
\begin{figure}[]
	\centering
	\includegraphics[width=1.0\linewidth]{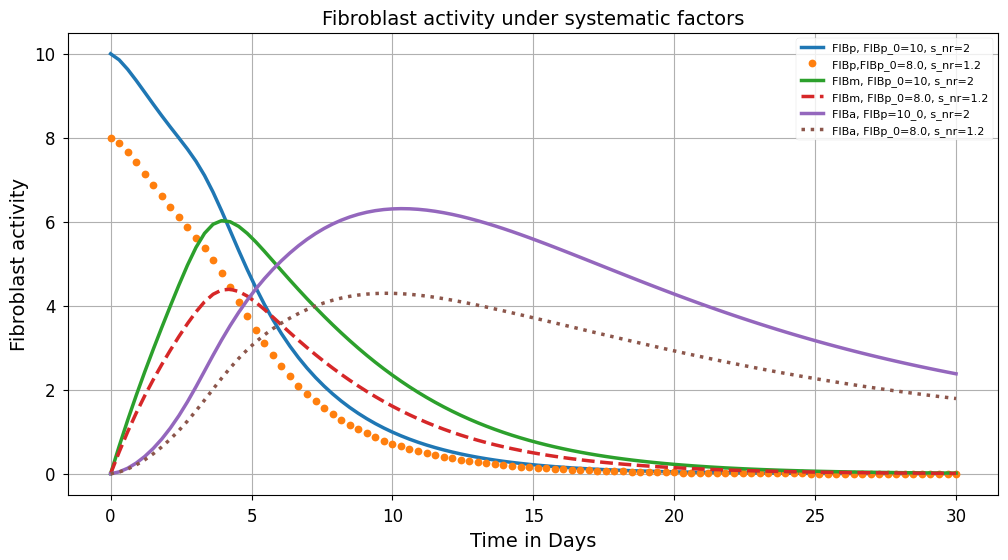}
	\caption{\begin{small}
			Activity of fibroblast cells in presence of systemic factors. Setting  $WS_0=0.6, FIB_\text{p}=8, FIB_\text{m}=0, FIB_\text{a}=0, COL=0, N=0, P=0.5$
	\end{small}}
	
	\label{Fig3.12}
\end{figure} In Figure \ref{Fig3.12}, it is observed that a reduced activity of proliferating fibroblast cells will reduce activity of migrating fibroblast cells and also affect the activity of active fibroblast cells as compared to the normal wound healing process. These reduced cell activities result in a delayed wound healing process. Figure \ref{Fig3.13} shows the increased activity of inflammatory cells due to the altered production of resting inflammatory cells by the body. This shows a prolonged inflammation due to a reduction in the production of resting inflammatory cells and their mediators due to systemic factors.
\begin{figure}[]
	\centering
	\includegraphics[width=1.0\linewidth]{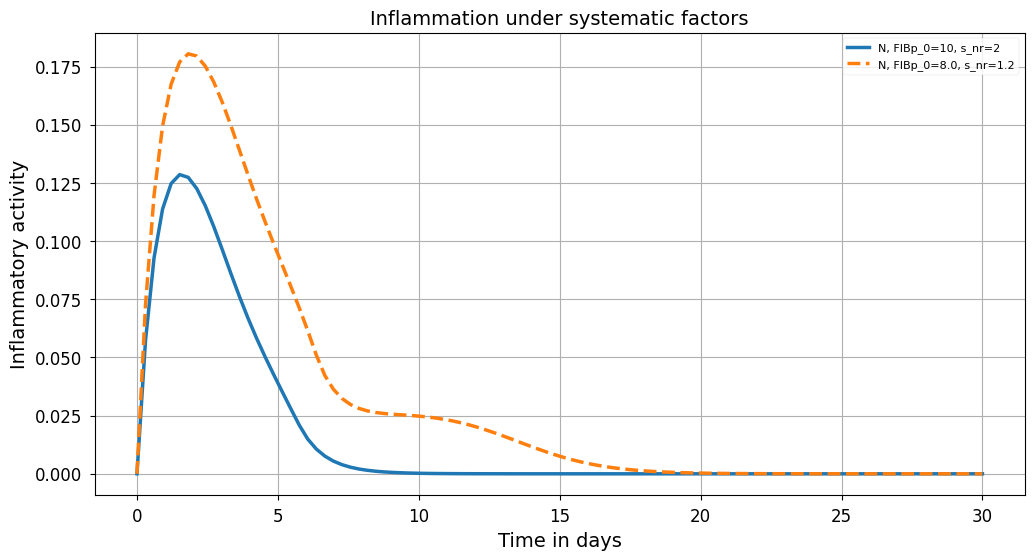}
	\caption{\begin{small}
			The activity of fibroblast cells in the presence of systemic factors. Setting  $WS_0=0.6, FIB_\text{p}=8, FIB_\text{m}=0, FIB_\text{a}=0, COL=0, N=0, P=0.5$
	\end{small}}
	
	\label{Fig3.13}
\end{figure}

\section{Conclusion}
This study builds upon the mathematical model developed by \cite{segal2012differential} to investigate the effects and contributions of different factors in the wound healing process. Through simulations and analysis, we have gained valuable insights into the effects of different factors on wound healing dynamics.  

Specifically, our findings highlight the crucial role of oxygen in controlling pathogen growth and enhancing the activity of inflammatory cells, thereby supporting optimal wound recovery. Additionally, we have identified the disruptive impact of increased pathogen activity in unclean wounds on the normal healing process, aligning with the classification of local factors affecting wound healing by \cite{guo2010factors}. The model has also been used to show the activity of cells in a chronic wound. Our model predicts that low oxygen levels or heightened pathogen activity may lead to chronic wounds, which can facilitate the entry of pathogens into the bloodstream, resulting in sepsis, a life-threatening condition.Our model predicts that low oxygen levels or increased pathogen activity may lead to chronic wounds, which could facilitate the entry of pathogens into the bloodstream, resulting in sepsis, a life-threatening condition. 

A sensitivity analysis using the Latin Hypercube Sampling method is carried out to identify the impact of individual parameters with respect to proliferating fibroblast cells $(FIB_p)$, migrating fibrolast cells $(FIB_m)$ active fibrobroblast cells $(FIB_p)$) and with respect to the threshold ($R_w$; a threshold which when less than unity guarantees the existence of a wound free equilibrium). It is worth noting that across all the sensitivity analyses,  our results depicted that the contact inhibition constant of FIBp reproduction by collagen $(C_{\infty})$, rate of production of collagen by active fibroblast cells $(k_{cf})$, and the fibroblast decay rate $(\mu_{fib})$ were most sensitive parameters. 

Furthermore, our study explores how various cell types respond to systemic factors like age, stress, alcohol consumption, and smoking. We observed reduced activity of migrating and active fibroblast cells due to impaired fibroblast cell proliferation. Additionally, there is an extended inflammatory period resulting from a decreased production of resting inflammatory cells and their mediators. Several studies have considered the effect of local factors such as oxygenation in the healing process. Unique to this work is the inclusion of systemic factors from Figure \ref{Fig3.12} and Figure \ref{Fig3.13} to describe cell behaviors in the wound healing process, the sensitivity analysis for the whole parameter space on fibroblast cells and we have successfully determined important parameters to the activity of fibroblast cells and the mathematical analysis to determine the wound-free equilibrium condition. 

This study provides valuable insights into the wound healing process. By considering the interplay between local factors and systemic factors, and important parameters in the wound healing process, we enhance our understanding of wound healing dynamics, which has important implications for effective wound management and the development of clinical therapies. It is crucial to recognize that the model's size makes it challenging to gather longitudinal data for calibration. However, as we look toward future directions, integrating empirical data could improve the relevance of our findings in clinical environments.

\section*{Data availability statement}
Data availability not applicable to this article as no datasets were generated or analysed during the current study.

\section*{CRediT authorship contribution statement}

\textbf{Alinafe Maenje:} Designed the study, developed the model, proved the mathematical properties, computed the numerical simulations, and wrote the original draft.
\textbf{Joseph Malinzi:} Supervised the study and reviewed the final draft.

\section*{Declaration of competing interests}

The authors declare that they have no known competing financial interests or personal relationships that could have appeared to influence the work reported in this paper.

\newcommand{\printdoi}[1]{
	\if\relax\detokenize{#1}\relax
	\else
	\newline DOI: \url{https://doi.org/#1}
	\fi
}

\bibliographystyle{apacite}
\bibliography{References}      

\end{document}